\documentclass{lmcs}

\usepackage[utf8]{inputenc}

\usepackage[sort]{natbib}
\usepackage{listings}
\lstdefinelanguage{Cext}[]{C}{morekeywords={boolean,true,false,assume,assert}}
\usepackage{amsthm}
\newtheorem{theorem}{Theorem}
\newtheorem{lemma}[theorem]{Lemma}

\usepackage{ifpdf,color,graphicx} 
\ifpdf

\else

\fi

\usepackage{alltt,verbatim}
\usepackage{amsfonts,amssymb,stmaryrd}

\usepackage{float} 
\usepackage{url}

\usepackage[pdfauthor={David Monniaux},
  pdfkeywords={abstract interpretation, quantifier elimination, program analysis},
  pdftitle={Automatic Modular Abstractions for Template Numerical Constraints}]{hyperref}
\usepackage{algorithm,algorithmic}

\title{Automatic Modular Abstractions for Template Numerical Constraints}

\date{May 26, 2010}
\author{David Monniaux}
\address{CNRS / VERIMAG\\Centre \'Equation\\2, avenue de Vignate\\38610 Gi\`eres cedex\\France}
\email{David.Monniaux@imag.fr}
\thanks{VERIMAG is a joint research laboratory of CNRS, Universit\'e Joseph Fourier and Grenoble-INP}{
\thanks{\parbox{2em}{\includegraphics[width=1.7em]{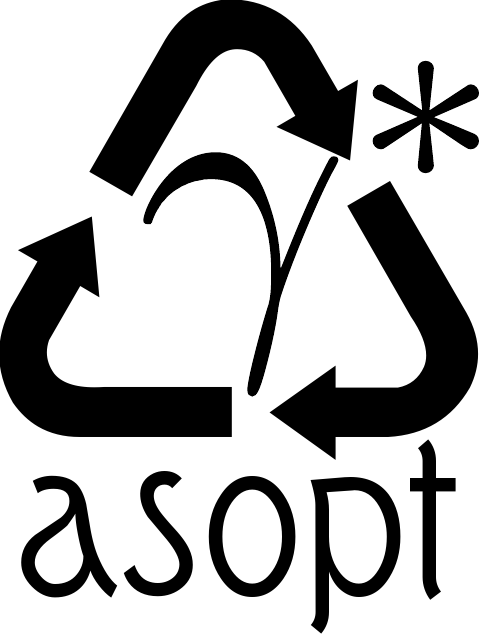}} This work was partially funded by the ``ASOPT'' project of the \emph{Agence nationale de la recherche} (ANR)}

\newcommand{\true}{\textrm{true}}
\newcommand{\false}{\textrm{false}}

\newcommand{\eabs}{\varepsilon_{\text{abs}}}
\newcommand{\elast}{\varepsilon_{\text{last}}}
\newcommand{\erel}{\varepsilon_{\text{rel}}}

\newcommand{\bbQ}{\mathbb{Q}}

\newcommand{\bbZ}{\mathbb{Z}}
\newcommand{\bbR}{\mathbb{R}}

\newcommand{\parts}[1]{\mathcal{P}(#1)}
\newcommand{\sem}[1]{\llbracket #1 \rrbracket}
\newcommand{\definedAs}{\stackrel{\vartriangle}{=}}
\newcommand{\lfp}{\mathop{\text{lfp}}}

\newcommand{\abstr}[1]{#1^\sharp}

\makeatletter
\let\@ORGmakecaption\@makecaption
\long\def\@makecaption#1#2{\@ORGmakecaption{#1}{#2}\vskip\belowcaptionskip\relax}
\makeatother

\begin{document}
\begin{abstract}
We propose a method for automatically generating abstract transformers for static analysis by abstract interpretation. The method focuses on linear constraints on programs operating on rational, real or floating-point variables and containing linear assignments and tests. 
Given the specification of an abstract domain, and a program block, our method automatically outputs an implementation of the corresponding abstract transformer. It is thus a form of program transformation.

In addition to loop-free code, the same method also applies for obtaining least fixed points as functions of the precondition, which permits the analysis of loops and recursive functions.

The motivation of our work is data-flow synchronous programming languages, used for building control-command embedded systems, but it also applies to imperative and functional programming.

Our algorithms are based on quantifier elimination and symbolic manipulation techniques over linear arithmetic formulas. We also give less general results for nonlinear constraints and nonlinear program constructs. 
\end{abstract}

\maketitle

\section{Introduction}
Program analysis consists in deriving properties of the possible executions of a program from an algorithmic processing of its source or object code. Example of interesting properties include: ``the program always terminates''; ``the program never executes a division by zero''; ``the program always outputs a well-formed XML document''; ``variable \lstinline|x| always lies between 1 and 3''. There has been a considerable amount of work done since the late 1970s on \emph{sound} methods of program analysis --- that is, methods that produce results that are guaranteed to hold for all program executions, as opposed to \emph{bug finding} methods such as program testing, which cannot provide such guarantees in general.

Static analysis by \emph{abstract interpretation} is one of the various approaches to sound program analysis. Grossly speaking, abstract interpretation casts the problem of obtaining supersets of the set of reachable states of programs into a problem of finding fixed points of certain monotone operators on certain ordered sets, known as \emph{abstract domains}. When dealing with programs operating on arithmetic values (integer or real numbers, or, more realistically, bounded integers and floating-point values), these sets are often defined by numerical constraints, and ordered by inclusion. One may, for instance, attempt to compute, for each program point and each variable, an interval that is guaranteed to contain all possible values of that variable at that point. The problem is of course how to compute these fixed points. Obviously, the smaller the intervals, the better, so we would like to compute them as small as possible. Ideally, we would like to compute the least fixed point, that is, the least inductive invariant that can be expressed using intervals.

The purpose of this article is to expose how to compute such least fixed points exactly, at least for certain classes of programs and certain abstract domains. Specifically, we consider programs operating over real variables using only linear comparisons (e.g. $x + 2y \leq 3$ but not $x^2 \leq y$), and abstract domains defined using a finite number of linear constraints $\sum_i a_i v_i \leq C$, where the $a_i$ are fixed coefficients, $C$ is a parameter (whose computation is the goal of the analysis) and the $v_i$ are the program variables. Such domains evidently include the intervals, where the constraints are of the form $v_i \leq C$ and $-v_i \leq C$.

Not only can we compute such least fixed points exactly if all parameters are known, but we can also deal with the case where some of the parameters are unknowns, in which case we obtain the parameters of the least fixed point as explicit, algorithmic, functions of the unknowns. We can thus generate, once and for all, the \emph{abstract transformers} for blocks of code: that is, those that map the parameters in the precondition of the block of code to a suitable postcondition. For instance, in the case of interval analysis, we derive explicit functions mapping the bounds on variables at the entrance of a loop-free block to the tightest bounds at the outcome of the block, or to the least inductive interval invariant of a loop. This allows \emph{modular analysis}: it is possible to analyse functions or other kind of program modules (such as nodes in synchronous programming) in isolation of the rest of the code.

A crucial difference with other methods, even on loop-free code, is that we derive the optimal --- that is, the most precise --- abstract transformer for the whole sequence. In contrast, most static analyses only have optimal transformers for individual instructions; they build transformers for whole sequences of code by composition of the transformers for the individual instructions. Even on very simple examples, the optimal transformer for a sequence is not the composition of the optimal transformers of the individual instructions. Furthermore, most other methods are forced to merge the information from different execution traces at the end of ``if-then-else'' and other control flow constructs in a way that loses precision. In contrast, our method distinguishes different paths in the control flow graph as long as they do not go through loop iteration.

Our approach is based on \emph{quantifier elimination}, a technique operating on logical formulas that transforms a formula with quantifiers into an equivalent quantifier-free formula: for instance, $\forall x~(x \geq y \Rightarrow x \geq 1)$ is turned into the equivalent $y \geq 1$. Essentially, we define the result that we would like to obtain using a quantified logical formula, and, using quantifier elimination and further formula manipulations, we obtain an algorithm that computes this result. Thus, one can see our method as automatically transforming a non executable \emph{specification} of the result of the analysis into an algorithm computing that result.
\medskip

In \S\ref{part:background}, we shall provide background about abstract interpretation and quantifier elimination. In \S\ref{part:template_linear_constraints_abstraction}, we shall provide the main results of the article, that is, how to derive optimal abstract transformers for template linear constraint domains. In \S\ref{part:extensions}, we shall investigate various extensions to that framework, still based on quantifier elimination in linear real arithmetic; e.g. how to deal with floating-point values. In \S\ref{part:control-flow}, we shall explain how to move from the single loop case to more complex control flow. In \S\ref{part:experiments}, we shall describe our implementation of the main algorithm and some limited extensions. In \S\ref{part:other_numerical_domains}, we shall investigate extensions of the same framework using quantifier elimination on other theories, namely Presburger arithmetic and the theory of real closed fields (nonlinear real arithmetic). In \S\ref{part:related}, we shall list some related work and compare our method to other relevant approaches. Finally, \S\ref{part:conclusion} concludes.

\emph{This article is based upon two conference articles \cite{Monniaux_POPL09,Monniaux_SAS07}.}

\section{Background}
\label{part:background}
In this section, we shall recall a few definitions, notations and results on program analysis by abstract interpretation, then quantifier elimination.

\subsection{Abstract interpretation}
\label{part:absint}

It is well-known that, in the general case, fully automatic program analysis is impossible for any nontrivial property.%
\footnote{This result, formally given within the framework of recursive function theory, is known as Rice's theorem~\citep[p.~34]{Rogers}\citep[corollary~B]{Rice_1953}. It is obtained by generalisation from Turing's halting theorem. Interpreted upon program semantics, the theorem states that the only properties of the denotational semantics of programs that can be algorithmically decided on the source code are the trivial properties: uniformly ``true'' or uniformly ``false''.}
Thus, all analysis methods must have at least one of the following characteristics:
\begin{itemize}
\item They may bound the memory size of the studied program, which then becomes a finite automaton, on which most properties are decidable. \emph{Explicit-state model-checking} works by enumerating all reachable states, which is tractable only  while \emph{implicit state model-checking} represents sets of states using data structures such as binary decision diagrams~\citep{ClarkeGrumbergPeled99}.
\item They may restrict the programming language used, making it not Turing-complete, so that properties become decidable. For instance, reachability in pushdown automata is decidable even though their memory size is unbounded~\citep{DBLP:conf/concur/BouajjaniEM97}.
\item They may restrict the class of properties expressed to properties of bounded executions; e.g., ``within the first 10000 steps of execution, there is no division by zero'', as in \emph{bounded model checking}~\citep{DBLP:journals/ac/BiereCCSZ03}.
\item They may be \emph{unsound} as proof methods: they may fail to detect that the desired property is violated. Typically, \emph{bug-finding} and testing programs are in that category, because they may fail to detect a bug. Some such analysis techniques are not based on program semantics, but rather on finding patterns in the program syntax~\citep{ASTREE_TASE07}.
\item They may be \emph{incomplete} as proof methods: they may fail to prove that a certain property holds, and report spurious violations. Methods based on abstraction fall in that category.
\end{itemize}

\emph{Abstraction} by over-approximation consists in replacing the original problem, undecidable or very difficult to decide, by a simpler ``abstract'' problem whose behaviours are guaranteed to include all behaviours of the original problem.

\lstset{language=Cext}
\begin{figure}\noindent%
\begin{minipage}{0.45\textwidth}
\begin{lstlisting}[caption={Original program}]
int x, y;
bool a;
if (x < y) a = false;
\end{lstlisting}
\end{minipage}\hfill
\begin{minipage}{0.45\textwidth}
\begin{lstlisting}[caption={Boolean program}]
bool a;
if (nondet()) a = false;
\end{lstlisting}
\end{minipage}

\caption{Transformation of a program into a Boolean program by erasing the numeric part and replacing tests over numerical quantities by nondeterministic choice (\lstinline|nondet()| nondeterministically returns true or false).}
\label{fig:Boolean_transform}
\end{figure}
\bigskip

An example of an abstraction is to erase from the program all constructs dealing with numerical and pointer types (or replacing them with nondeterministic choices, if their value is used within a test), keeping only Boolean types (Fig.~\ref{fig:Boolean_transform}). Obviously, the behaviours of the resulting program encompass all the behaviours of the original program, plus maybe some extra ones.

Further abstraction can be applied to this Boolean program: for instance, the ``3-value logic'' abstraction~\citep{DBLP:conf/cav/RepsSW04}
which maps any input or output variable to an abstract parameter taking its value in a 3-element set: ``is 0'', ``is 1'', ``can be 0 or 1'';%
\footnote{For brevity, we identify ``false'' with 0 and ``true'' with 1.}
for practical purposes it may be easier to encode these values using a couple of Booleans, respectively meaning ``can be 0'' and ``can be 1'', thus the abstract values $(1,0)$, $(0,1)$ and $(1,1)$. The abstract value $(0,0)$ obtained for any variable at a program point means that this program point is unreachable. Given a vector of input abstract parameters, one for each input variable of the program, the \emph{forward abstract transfer function} gives a correct vector of output abstract parameters, one for each output variable of the program. Quite obviously, in the absence of loops, it is possible to obtain a suitable forward abstract transfer function by applying simple logical rules to the Boolean function defined by the Boolean program. An effective implementation of the forward abstract transfer function can thus be obtained by a transformation of the source program.
\bigskip

For programs operating over numerical quantities, a common abstraction is \emph{intervals}.~\citep{CousotCousot76-1,Cousot_PhD} To each input $x$, one associates an interval $[x_{\min},x_{\max}]$, to each output $x'$ an interval $[x'_{\min},x'_{\max}]$. How can one compute the $(x'_{\min},x'_{\max})_{x \in V}$ bounds from the $(x_{\min},x_{\max})_{x \in V}$? The most common method is \emph{interval arithmetic}: to each elementary arithmetic operation, one attaches an abstract counterpart that gives bounds on the output of the operation given bounds on the inputs. For instance, if one knows $[a_{\min},a_{\max}]$ and $[b_{\min},b_{\max}]$, and $c=a+b$, then one computes $[c_{\min},c_{\max}]$ as follows: $c_{\min}=a_{\min}+b_{\min}$ and $c_{\max}=a_{\max}+b_{\max}$. If a program point can be reached by several paths (e.g. at the end of an if-then-else construct), then a suitable interval $[x_{\min},x_{\max}]$ can be obtained by a \emph{join} of all the intervals for $x$ at the end of these paths: $[a,b] \sqcap [c,d] = [\min(a,c),\max(b,d)]$.
Again, for a loop-free program, one can obtain a suitable effective forward abstract transfer function by a program transformation of the source code. 

\begin{figure}
\begin{minipage}{0.45\textwidth}
\begin{lstlisting}[mathescape=true,caption={Program computing zero}]
double x, y, z;
/* x lies in $[x_{\min},x_{\max}]$} */
y = x;
z = x-y;
\end{lstlisting}
\end{minipage}
\hfill
\begin{minipage}{0.45\textwidth}
\begin{lstlisting}[mathescape=true,caption={Interval transformer}]
$y_{\min} = x_{\min}$;
$y_{\max} = x_{\max}$;
$z_{\min} = x_{\min} - y_{\max}$;
$z_{\max} = x_{\max} - y_{\min}$;
\end{lstlisting}
\end{minipage}

\caption{The transformer for the interval domain obtained by composition of locally optimal abstract transformers is imprecise. For each statement (on the left) we use a corresponding optimal transformer (on the right), but the composition of these transformers is not optimal. For the sake of simplicity, all variables are considered to be real numbers.}
\label{fig:imprecise_interval}
\end{figure}

The abstract transfer function defined by interval arithmetic is always correct, but is not necessarily the most precise. For instance, on example in Fig.~\ref{fig:imprecise_interval},
the best abstract transfer function maps any input range to $z_{\min}=z_{\max}=0$, since the output $z$ is always zero, while the one obtained by applying interval arithmetic to all program statements yields, in general, larger intervals.
The weakness of the interval domain on this example is evidently due to the fact that it does not keep track of relationships between variables (here, that $x=y$). \emph{Relational abstract domains} such as the \emph{octagons} \citep{mine:hosc06,Mine_AST_WCRE01,Mine_PhD} or \emph{convex polyhedra} \citep{CousotHalbwachs78,Halbwachs_PhD,PPL} address this issue. Yet, neither octagons nor polyhedra provide analyses that are guaranteed to give optimal results.

Consider the following program:
\begin{lstlisting}[caption={Undistinguished paths}]
int x;
if (x > 0) x= 1; else x= -1;
if (x == 0) x= 2;
\end{lstlisting}
Obviously, the second test can never be taken, since $x$ can only be $\pm 1$; however an interval, octagon or polyhedral analysis will conclude, after joining the informations for both branches of the first test, that $x$ lies in $[-1,1]$ and will not be able to exclude the case $x = 0$. The final invariant will therefore be $x \in [-1,2]$.

In contrast, our method, applied to this example, will correctly conclude that the optimal output invariant for $x$ is $[-1,1]$. In fact, our method yields the same result as considering the (potentially exponentially large) set of paths between the beginning and the end of the program, and for each path, computing the least output interval, then computing the join of all these intervals.

\begin{figure}
\begin{center}
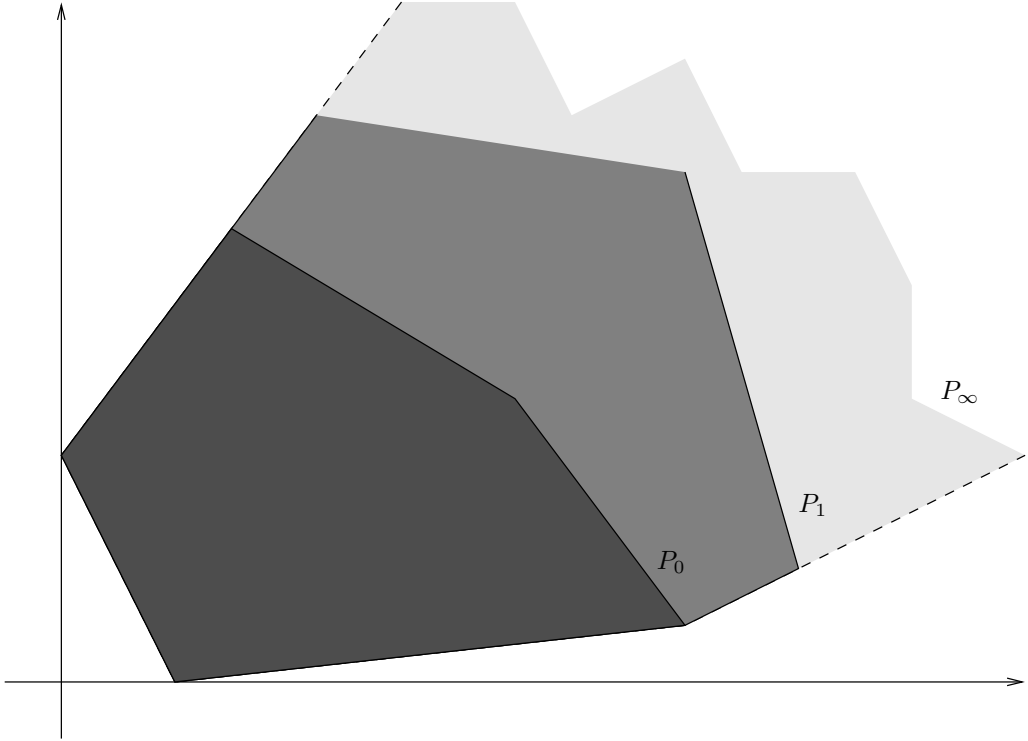
\end{center}

\caption{The standard widening on convex polyhedra \cite{Halbwachs_PhD,CousotHalbwachs78}, here demonstrated on polyhedra in dimension 2 (polygons). The widening operator observes the sets of reachable states $P_0$ and $P_1$ at two consecutive iterations, and keeps only the constraints (polyhedral faces, here polygon edges) that are stable across iterations. The resulting $P_\infty$ polyhedron is then proposed as an invariant.}
\label{fig:widening_polyhedra}
\end{figure}

We have so far left out programs containing loops; when programs contain loops or recursive functions, a central problem of program analysis is to find \emph{inductive invariants}. In the case of Boolean programs, given constraints on the input parameters, the set of reachable states can be computed exactly by model-checking algorithms; yet, these algorithms do not give a closed-form representation of the abstract transfer function mapping input parameters to output parameters for the 3-value abstraction. In the case of numerical abstractions such as the intervals, octagons or polyhedra, the most common way to find invariants is through the use of a \emph{widening operator}~\citep{CousotCousot76-1,CousotCousot_JLC92}.

Intuitively, widening operators observe the sets of reachable states after $N$ and $N+1$ loop iterations and extrapolate them to a ``candidate invariant''. For instance, the widening operator, observing a sequence of intervals $[0,1]$, $[0,2]$, $[0,3]$ may wish to try $[0,+\infty)$. See Fig.~\ref{fig:widening_polyhedra} for an example with the standard widening operator on convex polyhedra.

Let $u_0$ be the set of initial states of a loop, and let $\rightarrow_\tau$ be transition relation for this loop ($\sigma \rightarrow_\tau \sigma'$ means that $\sigma'$ is reachable in one loop step from $\sigma$). The set of reachable states at the head of the loop is the least fixed point of $f: u \mapsto u \cup \{ \sigma' \mid \exists \sigma \in u \land \sigma \rightarrow_\tau \sigma' \}$, which is obtained as the limit of the ascending sequence defined by $u_{n+1} = f(u_n)$. By abstract interpretation, we replace this sequence by an abstract sequence $\abstr{u}_n$ defined by $\abstr{u}_{n+1} = \abstr{f}(\abstr{u}_n)$, such that for any $n$, $\abstr{u}_n$ is an abstraction of $u_n$. If this sequence is stationary, that is, $\abstr{u}_{N+1}=\abstr{u}_N$ for some~$N$, then $\abstr{u}_N$ is an abstraction of the least fixed point of $f$ and thus of the least invariant of the transition relation $\tau$ containing~$u_0$.

When one uses a widening operator $\triangledown$, the function $\abstr{f}(\abstr{u}_n)$ is defined as  $\abstr{u}_n \triangledown \abstr{f}_p(\abstr{u}_n)$ where $\abstr{f}_p$ is an abstraction of $f$. The design of the widening operator ensures convergence of $\abstr{u}_n$ in finite time.
The exceptions to the use of widening operators are static analysis domains that satisfy the \emph{ascending condition}, such as the domain of linear equality constraints \citep{Karr76} and that of linear congruence constraints~\citep{Granger91}: with $\abstr{f}=\abstr{f}_p$ the sequence $\abstr{u}_n$ always becomes stable within finite time.

Again, widening operators provide correct results, but these results can be grossly over-approximated. Much of the literature on applied analysis by abstract interpretation discusses workarounds that give precise results in certain cases: \emph{narrowing} iterations \citep{CousotCousot76-1,Cousot_PhD}, widening ``up to'' \citep[\S3.2]{Halbwachs_CAV93}, ``delayed'' or with ``thresholds'' \citep{ASTREE_PLDI03}, etc.

A simple example of a program with one single variable where narrowing iterations fail to improve precision is Listing~\ref{ex:circularbuffer}.
\lstset{language=Cext}
\begin{lstlisting}[float,label=ex:circularbuffer,caption={Circular buffer}]
i = 0;
while (true) {
  if (nondet()) {
    i = i+1;
    if (i >= 10) i=0;
  }
}
\end{lstlisting}
This program is a much simplified version of a piece of code maintaining a circular buffer inside a large reactive control loop, where some piece of data is inserted only at certain loop iterations. We only kept the instructions relevant to the array index \lstinline|i| and abstracted away the choice of the loop iterations where data is inserted as nondeterministic \lstinline|nondet()|.

If we analyse this loop using intervals with the standard widening with a widening point at the head of the loop, we obtain the sequence $[0,0]$, $[0,1]$, $[0,2]$ and then widening to $[0,+\infty)$. Narrowing iterations then fail to increase the precision. The reason is that the transition relation for this loop includes the identity function (when \lstinline|nondet()| is false), thus the concrete function whose least fixed point defines the set of reachable states \citep{CousotCousot_JLC92} satisfies $X \subseteq \phi(X)$ for all~$X$ (in other word, $\phi$ is \emph{expansive}). Thus, once the widening operator overshoots the least fixed point, it can never recover it.

A similar problem is posed by:
\begin{lstlisting}
i = 0;
while (true) {
  i = i+1;
  if (i == 10) i=0;
}
\end{lstlisting}
The usual widening iterations overshoot to $[0,+\infty)$, and narrowing does not recover from there. Furthermore, this example illustrates how widening makes analysis \emph{non monotonic}: contrary to one could expect, having extra precision on the precondition of a program can result in worse precision for the inferred invariants. For instance, consider the above problem and replace the first line by \lstinline|assume(i>=0 && i<=9)|. Clearly, the resulting program is an abstraction of the above example, since it has strictly more behaviours (we allow $1,\dots,9$ as initial values for \lstinline|i|). Yet, the analysis of the loop will yield a more precise behaviour: the interval $[0,9]$ is stable and the analysis terminates immediately.

Both these very simple examples can be precisely analysed using \emph{widening up to}\citep[\S3.2]{Halbwachs_CAV93}, also known as \emph{widening with thresholds}~\cite[Sec.~7.1.2]{ASTREE_PLDI03}: a preliminary phase collects all constants to which \lstinline|i| is compared, and instead of widening to $+\infty$, one widens to the next larger such constant if it exists --- in this case, since \lstinline|x < 10| stands for $x \leq 9$, widening with threshold would widen to~$9$, which is the correct value. This approach is not general --- it fails if instead of the constant 10 we have some computed value. Of course, improvements are possible: for instance, one could analyse all the program up to this loop in order to prove that certain variables are constant, then use this information for setting thresholds for further loops. Yet, again, this is not a general approach.

This is the second problem that this article addresses: how to obtain, in general, optimal invariants for certain classes of programs and numerical constraints. Furthermore, our methods provide these invariants as functions of the parameters of the precondition of the loop; this is one difference with our proposal, which computes the best, thus, again, they provide effective, optimal abstract transfer functions for loop constructs.

\subsection{Quantifier elimination}
\label{part:qe}
Consider a set $A$ of atomic formulas. The set $U(A)$ of \emph{quantifier-free formulas} is the set of formulas constructed from $A$ using operators $\land$, $\lor$ and $\neg$; the set $Q(A)$ of \emph{quantified formulas} is the set of formulas constructed from $A$ using the above operators and the $\exists$ and $\forall$ quantifiers. Such formulas are thus trees whose leaves are the atomic formulas.
A \emph{literal} is an atomic formula or the negation thereof. The set of \emph{free variables} $FV(F)$ of a formula $F$ is defined as usual. A quantifier-free formula without variables is said to be \emph{ground}. A formula without free variables is said to be \emph{closed}; the \emph{existential closure} of a formula $F$ is $F$ with existential quantifiers for all free variables prepended; the \emph{universal closure} is the same with universal quantifiers. A quantifier-free formula is said to be in \emph{disjunctive normal form} (DNF) if it is a disjunction of conjunctions, that is, is of the form $(l_{1,1} \land \dots \land l_{1,n_1}) \lor \dots \lor (l_{m,1} \land \dots \land l_{m,n_m})$, and is said to be in \emph{conjunctive normal form} (CNF) if it is a conjunction of disjunctions. Any quantifier-free formula can be converted into CNF or DNF by application of the distributivity laws $(A \lor B) \land C \equiv (A \land C) \lor (B \land C)$ and $(A \land B) \lor C \equiv (A \lor C) \land (B \lor C)$, though better algorithms exist, such as ALL-SAT modulo theory~\citep{Monniaux_LPAR08}.

\subsubsection{Linear real inequalities}
\label{part:qf_lra}
Let $A$ be the set of linear inequalities with integer or rational coefficients over a set of variables $V$. By elementary calculus, such inequalities can be equivalently written in the following forms: $l(v_1, v_2, \dots) \geq C$ or $l(v_1, v_2, \dots) > C$, with $l$ a linear expression with integer coefficients over~$V$ and $C$ a constant. Let us first consider the theory of \emph{linear real arithmetic} (LRA): models of a formula $F$ are mappings from $F$ to the real field~$\bbR$, and notions of equivalence and satisfiability follow. Note that satisfiability and equivalence are not affected by taking models to be mappings from $F$ to the rational field~$\bbQ$. Deciding whether a LRA formula is satisfiable is, again, a NP-complete problem known as \emph{satisfiability modulo theory} (SMT) of real linear arithmetic. Again, practical implementations, known as SMT-solvers , are capable of dealing with rather large formulas; examples include Yices~\citep{DBLP:conf/cav/DutertreM06} and Z3~\citep{DBLP:conf/tacas/MouraB08}.%
\footnote{The yearly \href{http://www.smtcomp.org/}{SMT-COMP} competition has SMT-solvers compete on a large set of benchmarks. The \href{http://www/smtlib.org}{SMT-LIB} site \citep{SMTLIB} documents various theories amenable to SMT, including large libraries of benchmarks. \citet{Kroening_Strichman_08} is an excellent introductory material to the algorithms behind SMT-solving.}

The theory of linear real arithmetic admits quantifier elimination. For instance, the quantified formula $\forall x~ (x \geq y \implies x \geq 3)$ is equivalent to the quantifier-free formula $y \geq 3$.
Quantified linear real arithmetic formulas are thus decidable: by quantifier elimination, one can convert the existential closure of the formula to an equivalent ground formula, the truth of which is trivially decidable by evaluation.  The decision problem for quantified formulas over rational linear inequalities requires at least exponential time, thus quantifier elimination is at least exponential.~\cite[\S 7.4]{BradleyManna07}. \cite[Th.~3]{Fischer_Rabin_exponential_74} \citet{Weispfenning88} discusses complexity issues in more detail.

Again, most quantifier elimination algorithms proceed by induction over the structure of the formula, and thus begin by eliminating the innermost quantifiers, progressively replacing branches of the formula containing quantifiers by quantifier-free equivalent branches. By application of the equivalence $\forall x~F \equiv \neg \exists x~\neg F$, one can reduce the problem to eliminating existential quantifiers only. Consider now the problem of eliminating the existential quantifiers from $\exists~x_1 \dots x_n~F$ where $F$ is quantifier-free. We can first convert into DNF: $\exists x_1 \dots x_n~(C_1 \lor \dots \lor C_m)$ where the $C_i$ are conjunctions, then to the equivalent formula $(\exists~x_1 \dots x_n~C_1) \lor \dots \lor (\exists~x_1 \dots x_n~C_m)$. We thus have reduced the quantifier elimination problem for general formula to the problem of quantifier elimination for conjunctions of linear inequalities. Remark that, geometrically, this quantifier elimination amounts to computing the projection of a convex polyhedron along the dimensions associated with the variables $x_1, \dots, x_n$, with the original polyhedron and its projection being defined by systems of linear inequalities. The Fourier-Motzkin elimination procedure \citep[\S 5.4]{Kroening_Strichman_08} converts $\exists x_1 \dots x_n~C$ into an equivalent conjunction. This is what we refer to the ``conversion to DNF followed by projection approach''. This approach is good for quickly proving that the theory admits quantifier elimination, but it is very inefficient. We shall now see better methods.

\begin{figure}
\begin{center}
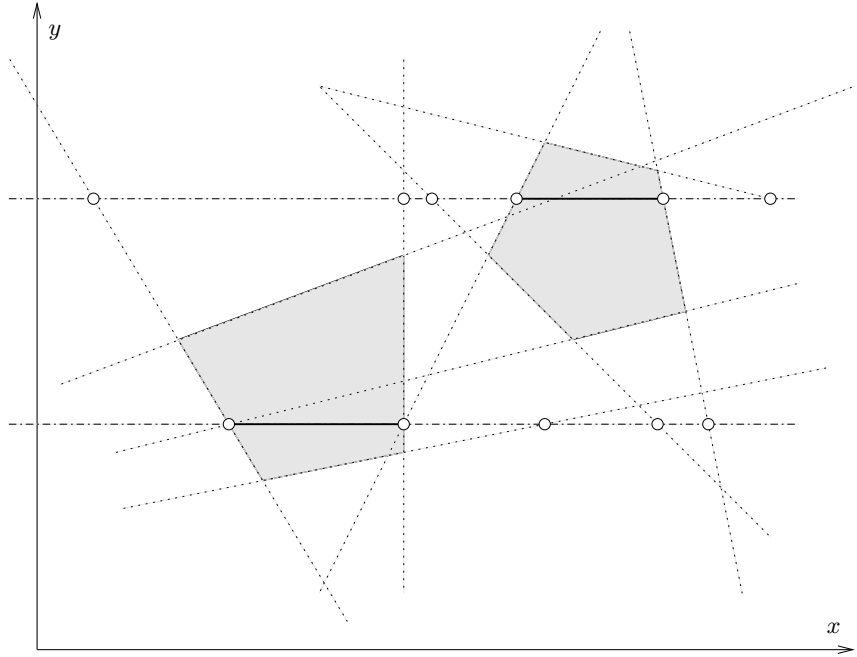
\end{center}

\caption{The gray zone $S$ is the set of $(x,y)$ solutions of formula $F$, whose atoms are the linear inequalities corresponding to the lines $\Delta$ drawn with dashes. For fixed $y=y_0$, the set of $x$ such that $F(x,y)$ is true is made of intervals whose ends lie within the set $I$ of intersections of the $y=y_0$ line with the lines in $\Delta$, drawn with a small circle. $y=y_0$ therefore has an intersection with $S$ if and only if a point in $I$, or an interval with both ends in $I \cup \{-\infty,+\infty\}$, lies within $S$. This condition can be tested using $x \rightarrow \pm \infty$ and all midpoints to intervals with both ends in $I$, as per \citet{FerranteRackoff75}, or, in addition to $I \cup \{ -\infty \}$, for any element of $I$ a point infinitesimally close to the right of it, as per \citet{LoosWeispfenning93}. Both methods exploit the fact that the coordinates of all points from $I$ (intersection of $y=y_0$ and a line from $\Delta$) can be expressed as affine linear functions of~$y_0$.}
\label{fig:qelim_substitution}
\end{figure}

\citet{FerranteRackoff75} proposed a substitution method \citep[\S 7.3.1]{BradleyManna07}\citep[\S 4.2]{NipkowIJCAR08}\citep{Weispfenning88}: a formula of the form $\exists x~F$ where $F$ is quantifier-free is replaced by an equivalent disjunction $F[e_1/x] \lor \dots \lor F[e_n/x]$, where the $e_i$ are affine linear expressions built from the free variables of~$\exists x~F$. Note the similarity to the naive elimination procedure we described for Boolean variables: even though the existential quantifier ranges over an infinite set of values, it is in fact only necessary to test the formula $F$ at a finite number of points (see Fig.~\ref{fig:qelim_substitution}).
The drawback of this algorithm is that $n$ is proportional to the square of the number of occurrences of $x$ in the formula; thus, the size of the formula can be cubed for each quantifier eliminated. \citet{LoosWeispfenning93} proposed a \emph{virtual substitution} algorithm%
\footnote{This method replaces $x$ by a formula that does not evaluate to a rational number, but to a sum of a rational number and optionally an infinitesimal $\varepsilon$, taken to be a number greater than 0 but less than any positive real; the infinitesimals are then erased by application of the rules governing comparisons. In practical implementations, one does both substitution and erasure of infinitesimals in one single pass, and infinitesimals never actually appear in formulas; thus the phrase \emph{virtual substitution}.}~\citep[\S 4.4]{NipkowIJCAR08}
that works along the same general ideas but for which $n$ is proportional to the number of occurrences of $x$ in the formula. Our benchmarks show that \citeauthor{LoosWeispfenning93}'s algorithm is generally much more efficient than \citeauthor{FerranteRackoff75}'s, despite the latter method being better known.~\citep{Monniaux_LPAR08,Monniaux_CAV10}

The main drawback of substitution algorithms is that the size of the formulas generally grows very fast as successive variables are eliminated. There are few opportunities for simplifications, save for replacing inequalities equivalent to $\false$ (e.g. $0 < 0$) by $\false$ and similarly for $\true$, then applying trivial rewrites (e.g. $\false \vee x \rightsquigarrow x$, $\false \land x \rightsquigarrow \false$). Our experience is that these algorithms tend to terminate with out-of-memory~\citep{Monniaux_LPAR08,Monniaux_CAV10}. \citet{DBLP:conf/tacas/SchollDPK09} have proposed using and-inverter graphs (AIGs) and SAT modulo theory (SMT) solving in order to simplify formulas and keeping their size manageable.

Another class of algorithms improve on the conversion to DNF then projection approach, by combining both phases: we proposed an eager algorithm based on this idea \citep{Monniaux_LPAR08}, then a lazy version, which computes parts of formulas as needed~\citep{Monniaux_CAV10}. Instead of syntactic conversion to DNF, we use a SMT solver to point to successive elements of the DNF, and instead of using Fourier-Motzkin elimination, which tends to needlessly blow up the size of the formulas, we use libraries for computations over convex polyhedra, which can compute a minimal constraint representation of the projection of a polyhedron given by constraints (that is, inequalities) --- which is the geometrical counterpart of performing elimination of a block of existentially quantified variables. Experiments have shown that such methods are generally competitive with substitution approaches, with different classes of benchmarks showing a preference for one of the two families of methods~\citep{Monniaux_CAV10}; for the kinds of problems we consider in this article, it seems that the SMT+projection methods are more efficient~\citep{Monniaux_LPAR08}.

\subsubsection{Presburger arithmetic}
\label{part:qf_lia}
The theory of \emph{linear integer arithmetic} (LIA), also known as Presburger arithmetic, has the same syntax for formulas, but another semantics, replacing rational numbers ($\bbQ$) by integers ($\bbZ$). Linear inequalities are then insufficient for quantifier elimination --- we also need congruence constraints: $\exists k~x=2k$ simplifies to $x \equiv 0 \pmod 2$.

Decision of formulas in Presburger arithmetic is doubly exponential, and thus quantifier elimination is very expensive in the worst case \cite{Fischer_Rabin_exponential_74}. \citet{Presburger29} provided a quantifier elimination procedure, but its complexity was impractical; \citet{Cooper72} proposed a better algorithm \citep[\S7.2]{BradleyManna07}; \citet{Pugh_Omega91} proposed the ``Omega test'' \citep[\S5.5]{Kroening_Strichman_08}.

The practical complexity of Presburger arithmetic is high. In particular, the formulas produced tend to be very complex, even when there exists a considerably simpler and ``understandable'' equivalent formula, as seen with experiments in \S\ref{part:presburger_abstraction}.

\citeauthor{Cooper72} and \citeauthor{Pugh_Omega91}'s procedures are very geometrical in nature. Integers, however, can also be seen as words over the $\{0,1\}$ alphabet, and sets of integers can thus be recognised by finite automata \citep[\S 8]{Perrin_HTCS}. Addition is encoded as a 3-track automaton recognising that the number on the third track truly is the sum of the numbers on the first two tracks; equivalently, this encodes subtraction. Existential quantifier elimination just removes some of the tracks, making transitions depending on bits read on that track nondeterministic. Negation is complementation (which can be costly, thus explaining the high cost of quantifier alternation). Multiplication by powers of two is also easily encoded, and multiplications by arbitrary constants can be encoded by a combination of additions and multiplications by powers of two.%
\footnote{This method embeds Presburger arithmetic into a stronger arithmetic theory, represented by the automata, then performs elimination over these automata. This partly explains why it is difficult to recover a Presburger formula from the resulting automaton.}
The same idea can be extended to real numbers written as their binary expansion, using automata on infinite words.

This leads to an interesting arithmetic theory, with two sorts of variables: reals (or rationals) and integers. This could be used to model computer programs, with integers for integer variables and reals for floating-point variables (if necessary by using the semantic transformations described in \S~\ref{part:float}). \citet{Boigelot_et_al_TOCL05} described a restricted class of $\omega$-automata sufficient for quantifier elimination. \citet{LIRA} implemented the \textsc{Lira} tool based on such ideas. Unfortunately, this approach suffers from two major drawbacks: the practical performances are very bad for purely real problems \citep{Monniaux_LPAR08}, and it is impossible to recover an arithmetical formula from almost all these automata. We therefore did not pursue this direction.

\subsubsection{Nonlinear real arithmetic}
\label{part:cad}
What happens if we do not limit ourselves to linear arithmetic, but also allow polynomials? Over the integers, the resulting theory is known as Peano arithmetic. It is well known that there can exist no decision procedure for quantified Peano arithmetic formulas.%
\footnote{One does not need the full language of quantified Peano formulas for the problem to become undecidable. It is known that there exists no algorithm that decides whether a given nonlinear Diophantine equation (a polynomial equation with integer coefficients) has solutions, and that deciding such a problem is equivalent to deciding Turing's halting problem. in other words, it is impossible to decide whether a formula $P(x_1,\dots,x_n)=0 \land x_1 \geq 0 \land \dots \land x_n \geq 0$ is satisfiable over the integers. See the literature on Hilbert's tenth problem, e.g. the book by \citet{Matiyasevich}.}
Since a quantifier elimination algorithm would turn a quantified formula without free variables into an equivalent ground formula, and ground formulas are trivially decidable, it follows that there can exist no quantifier elimination algorithm for this theory.

The situation is wholly different over the real numbers. The satisfiability or equivalence of polynomial formulas does not change whether the models are taken over the real numbers, the real algebraic numbers, or, for the matter, any \emph{real closed field}; this theory is thus known as the theory of real closed fields, or \emph{elementary algebra}. \citet{Tarski51_RAND,Tarski51} and \citet{Seidenberg54} showed that this theory admits quantifier elimination, but their algorithms had impractically high complexity. \citet{CAD_Collins75} introduced a better algorithm based on \emph{cylindrical algebraic decomposition}. For instance, quantifier elimination on $\exists x~ax^2+bx+c=0$ by cylindrical algebraic decomposition yields
\begin{multline}
\left(c<0\land \left(\left(b<0\land a\geq \frac{b^2}{4 c}\right)\lor \
(b=0\land a>0)\lor \left(b>0\land a\geq \frac{b^2}{4 \
c}\right)\right)\right)\lor c=0\lor\\
\left(c>0\land \left(\left(b<0\land \
a\leq \frac{b^2}{4 c}\right)\lor (b=0\land a<0)\lor \left(b>0\land a\leq \
\frac{b^2}{4 c}\right)\right)\right)
\end{multline}
Note the cylindrical decomposition: first, there is a case disjunction according to the values of $c$, then, for each disjunct for $c$, a case disjunction for the value of $b$; more generally, cylindrical algebraic decomposition builds a tree of case disjunctions over a sequence of variables $v_1, v_2, \dots$ , where the guard expressions defining the cases for $v_i$ can only refer to $v_1, \dots, v_i$. This decomposition only depends on the polynomials inside the formula and not on its Boolean structure, and computing it may be very costly even if the final result is simple. This is the intuition why despite various improvements \citep{CAD_1998,Basu_Pollack_Roy_algoreal_2006} the practical complexity of quantifier elimination algorithms for the theory of real closed fields remain high. The theoretical space complexity is doubly exponential \cite{Davenport_Heintz88,Brown_Davenport_ISSAC07}. This is why our results in \S\ref{sec:real_polynomial_abstraction} are of a theoretical rather than practical interest.

Minimal extensions to this formula language may lead to undecidability. This is for instance the case when one adds trigonometric functions: it is possible to define $\pi$ as the least positive zero of the sine, then define the set of integers as the numbers $k$ such that $\sin(k\pi)=0$, and thus one can encode Peano arithmetic formulas into that language~\citep{Anai_Weispfenning_00}. Also, naive restrictions of the language do not lead to lower complexity. For instance, limiting the degree of the polynomials to two does not make the problem simpler, since formulas with polynomials of arbitrary degrees can be encoded as formulas with polynomials of degree at most two, simply by introducing new variables standing for subterms of the original polynomials. For instance, $\exists x~ ax^3+bx^2+cx+d=0$ can be encoded, using Horner's form for the polynomial, as $\exists x \exists y \exists z ~ z= ax+b \land y= zx +c \land yx + d=0$. Certain stronger restrictions may however work; for instance, if the variables to be eliminated occur only linearly, then one can adapt the substitution methods described in~\S\ref{part:qf_lra}.

\section{Optimal Abstraction over Template Linear Constraint Domains}
\label{part:template_linear_constraints_abstraction}

When applying abstract interpretation over domains of linear constraints, such as octagons \cite{mine:hosc06,Mine_PhD,Mine_AST_WCRE01}, one generally applies a widening operator, which may lead to imprecisions. In some cases, \emph{acceleration} techniques leading to precise results can be applied \citep{DBLP:conf/sas/GonnordH06,Gonnord_PhD}: instead of attempting to extrapolate a sequence of iterates to its limit, as does widening, the exact limit is computed. In this section, we describe a class of constraint domains and programs for which abstract transfer functions of loop-free codes and of loops can be exactly computed; thus the \emph{optimality}. Furthermore, the analysis outputs these functions in closed form (as explicit expressions combining linear expressions and functional if-then-else constructs), so the result of the analysis of a program fragment can be stored away for future use; thus the \emph{modularity}. Our algorithms are based on quantifier elimination over the theory of real linear arithmetic~(\S\ref{part:qe}).

\subsection{Template Linear Constraint Domains}
\label{part:template_linear_constraints}
Let $F$ be a formula over linear inequalities. We call $F$ a domain definition formula if the free variables of $F$ split into $n$ \emph{parameters} $p_1,\dots,p_n$ and $m$ \emph{state variables} $s_1,\dots,s_m$. We note $\gamma_F: \bbQ^n \rightarrow \parts{\bbQ^m}$ defined by $\gamma_F(\vec{p}) = \{ \vec{s} \in \bbQ^m \mid (\vec{p}, \vec{s}) \models F \}$. As an example, the interval abstract domain for 3 program variables $s_1, s_2, s_3$ uses 6 parameters $m_1, M_1, m_2, M_2, m_3, M_3$; the formula is $m_1 \leq s_1 \leq M_1 \land m_2 \leq s_2 \leq M_2 \land m_3 \leq s_3 \leq M_3$.

In this section, we focus on the case where $F$ is a conjunction $L_1(s_1, \dots, s_m) \leq p_1 \land \dots \land L_n(s_1, \dots, s_m) \leq p_n$ of linear inequalities whose left-hand side is fixed and the right-hand sides are parameters. Such conjunctions define the class of \emph{template linear constraint domains}~\citep{Sankaranarayana+others/05/Scalable}. Particular examples of abstract domains in this class are:
\begin{itemize}
\item the intervals (for any variable $s$, consider the linear forms $s$ and $-s$); because of the inconvenience of talking intervals of the form $[-a,b]$, we shall often taking them of the form $[x_{\min},x_{\max}]$, with the optimal value for $x_{\min}$ being a greatest lower bound instead of a least upper bound;
\item the difference bound matrices (for any variables $s_1$ and $s_2$, consider the linear form $s_1-s_2$);
\item the octagon abstract domain (for any variables $s_1$ and $s_2$, distinct or not, consider the linear forms $\pm s_1 \pm s_2$)~\cite{Mine_AST_WCRE01}
\item the octahedra (for any tuple of variables $s_1, \dots, s_n$, consider the linear forms $\pm s_1 \dots \pm s_n$).~\cite{Clariso_Cortadella_SAS2004}
\end{itemize}

Remark that $\gamma_F$ is in general not injective, and thus one should distinguish the \emph{syntax} of the values of the abstract domain (the vector of parameters~$\vec{p}$) and their \emph{semantics} $\gamma_F(\vec{p})$. As an example, if one takes $F$ to be $s_1 \leq p_1 \land s_2 \leq p_2 \land s_1+s_2 \leq p_3$, then both $(p_1,p_2,p_3)=(1, 1, 2)$ and $(1, 1, 3)$ define the same set for state variables $s_1$ and $s_2$. If $\vec{u} \leq \vec{v}$ coordinate-wise, then $\gamma_F(\vec{u}) \subseteq \gamma_F(\vec{v})$, but the converse is not true due to the non-uniqueness of the syntactic form.

Take any nonempty set of states $W \subseteq \bbQ^m$. Take for all $i=1,\dots,m$: $p_i = \sup_{\vec{s} \in W} L_i(\vec{s})$. Clearly, $W \subseteq \gamma_F(p_1, \dots, p_m)$, and in fact $\vec{p}$ is such that $\gamma_F(\vec{p})$ is the least solution to this inclusion.
$p_i$ belongs in general to $\bbR \cup \{ +\infty \}$, not necessarily to $\bbQ \cup \{ +\infty \}$. (for instance, if $W = \{ s_1 \mid s_1^2 \leq 2 \}$ and $L_1=s_1$, then $p_1 = \sqrt{2}$). We have therefore defined a map $\alpha_F: \parts{\bbR^m} \rightarrow \{ \bot \} \cup (\bbR \cup \{ +\infty \})^n$, and $(\alpha_F, \gamma_F)$ form a \emph{Galois connection}: $\alpha_F$ maps any set to its best upper-approximation.%
\footnote{See e.g. \citep{CousotCousot_JLC92} for more information on Galois connection and their use in static analysis. Not all abstract interpretation techniques can be expressed within Galois connections. Indeed, there are abstract domains where there is not necessarily a best abstraction of a set of concrete states, e.g. the domain of convex polyhedra, which has no best abstraction for a disc, for instance. In this article, all abstractions, except the non-convex ones of \S~\ref{part:non-convex}, are through Galois connections.}
The fixed points of $\alpha_F \circ \gamma_F$ are the \emph{normal forms}; the normal form of $\abstr{x}$ is the minimal abstract element that stands for the same concrete set as~$\abstr{x}$.%
\footnote{In the terminology of some authors, these can be referred to as the \emph{reduced forms} or \emph{closed forms}, and the $\alpha_F \circ \gamma_F$ operation is a \emph{reduction} or \emph{closure}. For instance, in the octagon abstract domain, the closure $\alpha_F \circ \gamma_F$ is implemented by a variant of Floyd-Warshall shortest path~\cite{mine:hosc06,Mine_AST_WCRE01}.}
For instance, $s_1 \leq 1 \land s_2 \leq 1 \land s_1 + s_2 \leq 2$ is in normal form, while $s_1 \leq 1 \land s_2 \leq 1 \land s_1 + s_2 \leq 3$ is not.

\subsection{Optimal Abstract Transformers for Program Semantics}
\label{part:template_linear_constraints_transformers}
We shall consider the input-output relationships of programs with rational or real variables. We first narrow the problem to programs without loops and consider programs built from linear arithmetic assignments, linear tests, and sequential composition. 
Noting $a, b, \dots$ the values of program variables $\texttt{a}, \texttt{b}\dots$ at the beginning of execution and $a', b', \dots$ the output values, the semantics of a program $P$ is defined as a formula $\sem{P}$ such that $(a, b, \dots, a', b', \dots) \models P$ if and only if the memory state $(a', b', \dots)$ can be reached at the end of an execution starting in memory state $(a, b, \dots)$:
\begin{description}
\item[Arithmetic]
  $\sem{a:=L(a, b, \dots)+K}_F \definedAs a'=L(a, b, \dots)+K \land b'=b \land c'=c \land \dots$ where $K$ is a real (in practice, rational) constant and $L$ is a linear form, and $b,c,d\dots$ are all the variables except~$a$;

\item[Tests]
  $\sem{\texttt{if~} c \texttt{~then~} p_1 \texttt{~else~} p_2} \definedAs
   (c \land \sem{p_1}_F) \lor (\neg c \land \sem{p_2}_F)$;

\item[Non deterministic choice]
  $\sem{a:=\texttt{random}} \definedAs b'=b \land c'=c \land \dots$, for all variables except~$a$;

\item[Failure]
  $\sem{\texttt{fail}} \definedAs \textsf{false}$;

\item[Skip]
  $\sem{\texttt{skip}} \definedAs a'=a \land b'=b \land c'=c \land \dots$

\item[Sequence]
$\sem{P_1; P_2}_F \definedAs \exists a'', b'', \dots ~ f_1 \land f_2$ where
$f_1$ is $\sem{P_1}_F$ where $a'$ has been replaced by $a''$, $b'$ by $b''$ etc., 
$f_2$ is $\sem{P_2}_F$ where $a$ has been replaced by $a''$, $b$ by $b''$ etc.
\end{description} 

In addition to  linear inequalities and conjunctions, such formulas contain disjunctions (due to tests and multiple branches) and existential quantifiers (due to sequential composition).

Note that so far, we have represented the concrete denotational semantics \emph{exactly}. This representation of the transition relation using existentially quantified formulas is evidently as expressive as a representation by a disjunction of convex polyhedra (the latter can be obtained from the former by quantifier elimination and conversion to disjunctive normal form), but is more compact in general.

Consider now a domain definition formula $F \definedAs L_1(s_1, s_2, \dots) \leq p_1 \land \dots \land L_n(s_1, s_2, \dots) \leq p_n$ on the program inputs, with parameters $\vec{p}$ and free variables $\vec{s}$, and another $F' \definedAs L'_1(s'_1, s'_2, \dots) \leq p'_1 \land \dots \land L'_n(s'_1, s'_2, \dots) \leq p'_{n'}$ on the program outputs, with parameters $\vec{p'}$ and free variables $\vec{s'}$. Sound forward program analysis consists in deriving a \emph{safe post-condition} from a precondition: starting from any state verifying the precondition, one should end up in the post-condition. Using our notations, the \emph{soundness condition} is written 
\begin{equation}\label{eqn:soundness}
\forall \vec{s},\vec{s'}~ F \land \sem{P} \implies F'
\end{equation}
The free variables of this relation are $\vec{p}$ and $\vec{p'}$: the formula links the value of the parameters of the input constraints to admissible values of the parameters for the output constraints.
Note that this soundness condition can be written as a universally quantified formula, with no quantifier alternation. Alternatively, it can be written as a conjunction of correctness conditions for each output constraint parameter:
\begin{equation}\label{eqn:loopless_correctness}
C'_i \definedAs \forall \vec{s},\vec{s'}~ F \land \sem{P} \implies L'_i(\vec{s'}) \leq p'_i.
\end{equation}

Let us take a simple example: if $P$ is the program instruction $z:=x+y$, $F \definedAs x \leq p_1 \land y \leq p_2$, $F' \definedAs z \leq p'_1$, then $\sem{P} \definedAs z' = x+y$, and the soundness condition is
$\forall x, y, z'~ (x \leq p_1 \land y \leq p_2 \land z'=x+y \implies z' \leq p'_1)$. Remark that this soundness condition is equivalent to a formula without quantifiers $p'_1 \geq p_1 + p_2$, which may be obtained through quantifier elimination. Remark also that while any value for $p'_1$ fulfilling this condition is \emph{sound} (for instance, $p'_1=1000$ for $p_1=p_2=1$), only one value is \emph{optimal} ($p'_1=2$ for $p_1=p_2=1$). An optimal value for the output parameter $p'_i$ is defined by $O'_i \definedAs C'_i \land \forall q'_i~ (C'_i[q'_i / p'_i] \implies p'_i \leq q'_i)$. Again, quantifier elimination can be applied; on our simple example, it yields $p'_1 = p_1 + p_2$.

If there are $n$ input constraint parameters $p_1, \dots, p_n$, then the optimal value for each output constraint parameter $p'_i$ is defined by a formula $O'_i$ with $n+1$ free variables  $p_1, \dots, p_n, p'_i$:
\begin{equation}\label{eqn:loopless_optimality}
O'_i \definedAs C'_i \land \forall p''_i (C'_i[p''_i/p'_i] \Rightarrow p'_i \leq p'_i)
\end{equation}

\begin{lemma}\label{lem:loopless_correctness}
The formula $O'_i$ defined at Eq.~\ref{eqn:loopless_optimality}, using the correctness subformula $C'_i$ from Eq.~\ref{eqn:loopless_correctness}, defines $p'_i$ as the least possible value for the parameter of the constraint $L_i$ after executing the transition $\sem{p}$ from a state verifying constraints~$F$.
\end{lemma}

\begin{proof}
$O'_i$ explicitly defines the least possible value of $C'_i$. $C'_i$ explicitly defines all the acceptable values for parameter~$p'_i$ in the postcondition constraint.
\end{proof}

This formula defines a \emph{partial function} from $\bbQ^n$ to $\bbQ$, in the mathematical sense: for each choice of $p_1, \dots, p_n$, there exist at most a single $p'_i$. The values of $p_1, \dots, p_n$ for which there exists a corresponding $p'_i$ make up the \emph{domain of validity} of the abstract transfer function. Indeed, this function is in general not defined everywhere; consider for instance the program:
\lstset{language=Cext}
\begin{lstlisting}
if (x >= 10) { y = nondeterministic_choice_in_all_reals; }
else { y = 0; }
\end{lstlisting}
If $F = x \leq p_1$ and $F' = y \leq p'_1$, then $O'_1 \equiv p_1 < 10 \land p'_1=0$, and the function is defined only for $p_1 < 10$, with constant value~$0$.

At this point, we have a characterisation of the optimal abstract transformer corresponding to a program fragment $P$ and the input and output domain definition formulas as $n$ formulas (where $n$ is the number of output parameters) $O'_i$ each defining a function (in the mathematical sense) mapping the input parameters $\vec{p}$ to the output parameter~$p'_i$.

Another example: the absolute value function $y:=|x|$, again with the interval abstract domain. The semantics of the operation is $(x \geq 0 \land y=x) \lor (x < 0 \land y=-x)$; the precondition is $x \in [x_{\min}, x_{\max}]$ and the post-condition is $y \in [y_{\min}, y_{\max}]$.
Acceptable values for $(y_{\min},y_{\max})$ are characterised by
formula
\begin{equation}
C' \definedAs
\forall x~x_{\min} \leq x \leq x_{\max} \implies y_{\min} \leq |x| \leq y_{\max}
\end{equation}
The optimal value for $y_{\max}$ is defined by $O' \definedAs C' \land \forall y'_{\max}~ C'[y'_{\max}/y_{\max}] \allowbreak\implies\allowbreak y_{\max} \leq y'_{\max}$. Quantifier elimination over this last formula gives as characterisation for the least, optimal, value for $y_{\max}$:
\begin{equation}\label{eqn:abs_dnf}
(x_{\min} + x_{\max} \geq 0 \land y_{\max}=x_{\max}) \lor
(x_{\min} + x_{\max} < 0 \land y_{\max} = -x_{\min}).
\end{equation}

We shall see in the next sub-section that such a formula can be automatically compiled into code (Listing~\ref{ex:optimalabs}).
\lstset{language=Cext}
\begin{lstlisting}[float,caption={Optimal transformer for $y_{\max}$, corresponding to program $y = |x|$ with $x_{\min} \leq x \leq x_{\max}$},label=ex:optimalabs]
if (xmin + xmax >= 0) {
  ymax = xmax;
} else {
  ymax = -xmin;
}
\end{lstlisting}

\subsection{Generation of the Implementation of the Abstract Domain}
\label{part:template_linear_constraints_gen}
Consider formula~\ref{eqn:abs_dnf}, defining an abstract transfer function.
On this disjunctive normal form we see that the function we have defined is \emph{piecewise linear}: several regions of the range of the input parameters are distinguished (here, $x_{\min} + x_{\max} < 0$ and $x_{\min} + x_{\max} \geq 0$), and on each of these regions, the output parameter $y_{\max}$ is a linear function of the input parameters. Given a disjunct (such as $y_{\max} = -x_{\min} \land x_{\min} + x_{\max} < 0$), the domain of validity of the disjunct can be obtained by existential quantifier elimination over the result variable (here $\exists y_{\max}~(y_{\max} = -x_{\min} \land x_{\min} + x_{\max} < 0)$). The union of the domains of validity of the disjuncts is the domain of validity of the full formula. The domains of validity of distinct disjuncts can overlap, but in this case, since $O'_i$ defines a partial function in the mathematical sense, that is, a relation $R(x,y)$ such that for any $x$ there is at most one $y$ such that $R(x,y)$, the functions defined by such disjuncts coincide on their overlapping domains of validity.

This suggests a first algorithm for conversion to an executable form:
\begin{enumerate}
\item Put $O'_i$ into quantifier-free, disjunctive normal form $C_1 \lor \dots \lor C_n$.
\item For each disjunct $C_j$, obtain the validity domain $V_j$ as a conjunction of linear inequalities; this corresponds to a projection of the polyhedron defined by $C_j$ onto variables $p_1,\dots,p_n$, parallel to $p'_i$. Then, solve $C_j$ for $p'_i$ (obtain $p'_i$ as a linear function $v_i$ of the $p_1, \dots, p_n$).
\item Output the result as a cascade of if-then-else and assignments.
\end{enumerate}

Consider now the example at the end of \S\ref{part:template_linear_constraints_transformers} and especially Formula~\ref{eqn:abs_dnf}, defining $y_{\max}$ as a function of $x_{\min}$ and $x_{\max}$:
$(x_{\min} + x_{\max} \geq 0 \land y_{\max}=x_{\max}) \lor
(x_{\min} + x_{\max} < 0 \land y_{\max} = -x_{\min})$.
Let us first take the first disjunct $C_1 \definedAs x_{\min} + x_{\max} \geq 0 \land y_{\max}=x_{\max}$. Its validity domain is $\exists y_{\max}~C_1$, that is, $x_{\min} + x_{\max} \geq 0$. Furthermore, on this validity domain, we can solve for $y_{\max}$ as a function of $x_{\min}$ and $x_{\max}$, and we obtain $y_{\max} = x_{\max}$. We therefore can print out the following test:
\lstset{language=Cext}
\begin{lstlisting}
if (xmin + xmax >= 0) {
  ymax = xmax;
}
\end{lstlisting}

Now take the second disjunct $C_2 \definedAs (x_{\min} + x_{\max} < 0 \land y_{\max} = -x_{\min})$. It can similarly be turned into another test, which we put in the ``else'' branch of the preceding one. We thus obtain:
\lstset{language=Cext}
\begin{lstlisting}
if (xmin + xmax >= 0) {
  ymax = xmax;
} else
if (xmin + xmax < 0) {
  ymax = -xmax;
}
\end{lstlisting}

Observe that in the above program, the second test is redundant. If we had been more clever, we would have realized that in the ``else'' branch of the first test, $x_{\min}+x_{\max} < 0$ always holds, and thus it is useless to test this condition. Furthermore, in an if-then-else cascade obtained with the above method, the same condition may have to be tested several times. We shall now propose an algorithm that constructs an if-then-else \emph{tree} such that no useless tests are generated.


\begin{algorithm}
\caption{$\textsc{ToITEtree}(F,p',T)$: turn a formula defining $p'$ as a function of the other free variables of $F$ into a tree of if-then-else constructs, assuming that $T$ holds.}
\label{alg:ToIfThenElseTree}

\begin{algorithmic}
\STATE $D (= C_1 \lor \dots \lor C_n) \gets \textsc{QElimDNFModulo}(\{\},F,T)$
\FORALL{$C_i \in D$}
  \STATE $P_i \gets \textsc{QElimDNFModulo}(p', F, T)$
\ENDFOR
\STATE $P \gets \textsc{Predicates}(P_1, \dots, P_n)$
\IF{$P = \emptyset$}
  \ENSURE $\exists p'~F$ is always true
  \STATE $O \gets \textsc{Solve}(D, p')$
\ELSE
  \STATE $K \gets \textsc{Choose}(P)$
  \STATE $O \gets \textsf{IfThenElse}(K, \textsc{ToITEtree}(F,p',T \land K),\allowbreak \textsc{ToITEtree}(F,p',T \land \neg K))$
\ENDIF
\end{algorithmic}
\end{algorithm}

The idea of the algorithm is as follows:
\begin{itemize}
\item Each path in the if-then-else tree corresponds to a conjunction $C$ of conditions (if one goes through the ``if'' branch of \texttt{if (a)} and the ``else'' branch of \texttt{if (b)}, then the path corresponds to $a \land \neg b$).
\item The formula $O'_i$ is simplified relatively to~$C$, a process that prunes out conditions that are always or never satisfied when $C$~holds.
\item If the path is deep enough, then the simplified formula becomes a conjunction. This conjunction, as a relation, is a partial function from the $p_1,\dots,p_n$ to the $p'_i$ we wish to compute. Again, we solve this conjunction to obtain $p'_i$ as an explicit function of the $p_1,\dots,p_n$. For instance, in the above example, we obtain $y_{\max}$ as a function of $x_{\min}$ and~$x_{\max}$.
\end{itemize}

We say that two formulas $A$ and $B$ are equivalent, denoted by $A \equiv B$, if they have the same models, and we say that they are equivalent modulo a third formula $T$, denoted by $A \equiv_T B$, if $A \land T \equiv B \land T$. Intuitively, this means that we do not care what happens where $F$ is false, thus the terminology ``don't care set'' sometimes found for~$\neg F$.

$\textsc{QElimDNFModulo}\allowbreak(\vec{v},F,T)$ is a function that, given a possibly empty vector of variables $\vec{v}$, a formula $F$ and a formula $T$, outputs a quantifier-free formula $F'$ in disjunctive normal form such that $F' \equiv_T \exists \vec{v}~F$ and no useless predicates are used. In \citep{Monniaux_LPAR08}, we have presented such a function as a variant of quantifier elimination.

We need some more auxiliary functions. $\textsc{Predicates}(F)$ returns the set of atomic predicates of~$F$. $\textsc{Solve}(C, p')$ solves a conjunction of inequalities $C$ for variable $p'$. If $C$ does not contain redundant inequalities, then it is sufficient to look for inequalities of the form $p' \geq L$ or $p' \leq L$ and output $p' = L$.%
\footnote{In practice, any library for convex polyhedra can provide a basis of the equalities implied by a polyhedron given by constraints, in other words a system of linear inequalities defining the affine span of the polyhedron. It is therefore sufficient to take that basis and keep any equality where $p'$~appears.} Finally $\textsc{Choose}(P)$ chooses any predicate in the set of predicates $P$; one good heuristic seems to be to choose the most frequent atomic predicate where $p'_i$ does not occur.

Let us take, as a simple example, Formula~\ref{eqn:abs_dnf}. We wish to obtain $y_{\max}$ as a function of $x_{\min}$ and $x_{\max}$, so in the algorithm \textsc{ToITEtree} we set $p' \definedAs y_{\max}$. $C_1$ is the first disjunct $x_{\min} + x_{\max} \geq 0 \land y_{\max}=x_{\max}$, $C_2$ is the second disjunct $x_{\min} + x_{\max} < 0 \land y_{\max} = -x_{\min}$. We project $C_1$ and $C_2$ parallel to $y_{\max}$, obtaining respectively $P_1 = (x_{\min}+x_{\max} \geq 0)$ and $P_2 = (x_{\min}+x_{\max} < 0)$. We choose $K$ to be the predicate $x_{\min}+x_{\max} \geq 0$ (in this case, the choice does not matter, since $P_1$ and $P_2$ are the negation of each other).
\begin{itemize}
\item The first recursive call to \textsc{ToITEtree} is made with $T \definedAs (x_{\min}+x_{\max} \geq 0)$. Obviously, $F \land T \equiv (y_{\max}=x_{\max}) \land T$ and thus $(\exists y_{\max} F) \land T \equiv T$.

$\textsc{QElimDNFModulo}(y_{\max}, F, T)$ will then simply output the formula ``true''. It then suffices to solve for $y_{\max}$ in $y_{\max}=x_{\max}$. This yields the formula for computing the correct value of $y_{\max}$ in the cases where $x_{\min}+x_{\max} \geq 0$.

\item The second recursive call is made with $T \definedAs (x_{\min}+x_{\max} < 0 $. The result is $y_{\max}=-x_{\min}$, the formula for computing the correct value of $y_{\max}$ in the cases where $x_{\min}+x_{\max} < 0$.
\end{itemize}
These two results are then reassembled into a single if-then-else statement, yielding the program at the end of \S\ref{part:template_linear_constraints_transformers}.

The algorithm terminates because paths of depth $d$ in the tree of recursive calls correspond to truth assignments to $d$ atomic predicates among those found in the domains of validity of the elements of the disjunctive normal form of $F$. Since there is only a finite number of such predicates, $d$ cannot exceed that number. A single predicate cannot be assigned truth values twice along the same path because the simplification process in $\textsc{QElimDNFModulo}$ erases this predicate from the formula.

This algorithm seems somewhat unnecessarily complex. It is possible that techniques based upon AIGs, performing simplification with respect to ``don't care sets'' \citep{DBLP:conf/tacas/SchollDPK09}, could also be used.

\subsection{Least Inductive Invariants}
\label{part:template_linear_constraints_invariants}

We have so far considered programs without loops. The analysis of program loops, as well as proofs of correctness of programs with loops using Floyd-Hoare semantics, is based upon the notion of \emph{inductive invariants}. A set of states $I$ is deemed to be an inductive invariant for a loop if it contains the initial state and it is stable by the loop iteration --- in other words, if it is impossible to move from a state inside $I$ to a state outside $I$ by one iteration of the loop. The intersection of all inductive invariants is also an inductive invariant --- in fact, it is the least inductive invariant.

Any property true over the least inductive invariant of a loop is true throughout the execution of that loop; for this reason, some authors call such properties \emph{invariants} (without the ``inductive'' qualifier), while some other authors call invariants what we call inductive invariants.

It would be interesting to be able to compute the least invariant (inductive or noninductive) within our chosen abstract domain; in other words, compute the least element in our abstract domain that contains the least inductive invariant of the loop or program. Unfortunately, this is in general impossible; indeed, doing so would entail solving the halting problem. Just take any program $P$, create a fresh variable, and consider the program that initialises $x$ to $0$, runs $P$, and then sets $x$ to $1$. Clearly, the least invariant interval for this program and variable $x$ is $[0,0]$ if $P$ does not terminate, and $[0,1]$ if it does terminate.

We thus settle for a simpler problem: find the \emph{least inductive invariant within our abstract domain}, that is, the least element in our abstract domain that is an inductive invariant. Note that even on very simple examples, this can be vastly different from computing the least invariant within the abstract domain. Take for instance the simple program of Fig.~\ref{fig:unstable}, which has a couple of real variables $(x,y)$ and rotates them by $45^\circ$ at each iteration. The $(x,y)$ couple always stays within the square $[-1,1]\times[-1,1]$, so this square is an invariant within the interval domain. Yet this square is not an inductive invariant, for it is not stable if rotated by $45^\circ$; in fact, the only inductive invariant within the interval domain is $(-\infty,+\infty)\times(-\infty,+\infty)$, which is rather uninteresting! Note that on the same figure, the octagon abstract domain would succeed (and produce a regular octagon centered on $(0,0)$).

\begin{figure}
\begin{center}
\includegraphics[width=0.5\columnwidth]{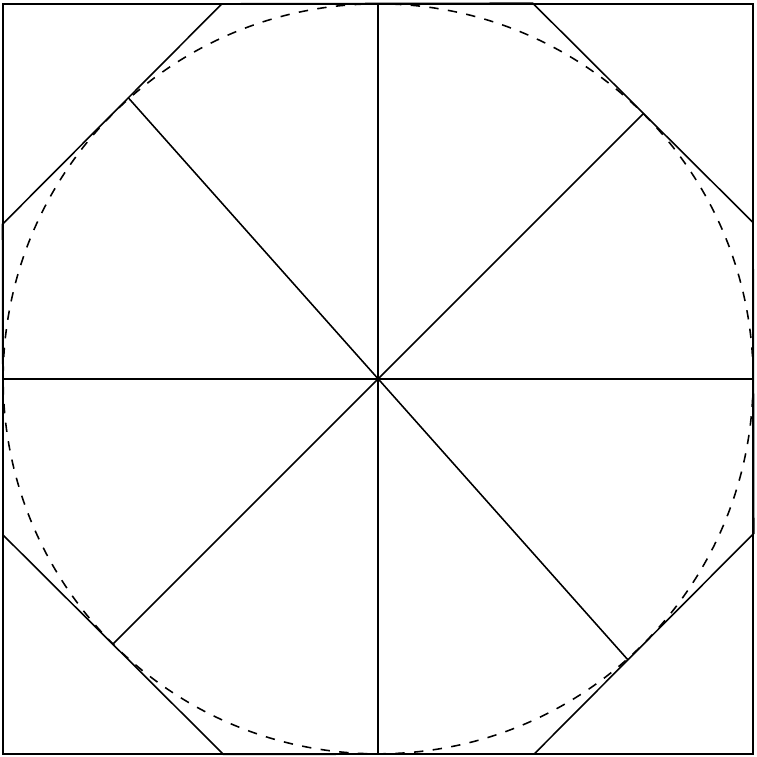}
\end{center}
\caption{The least fixed point representable in the domain ($\lfp(\alpha \circ f \circ \gamma)$) is not necessarily the least approximation of the least fixed point  ($\alpha(\lfp f)$) inside the abstract domain. For instance, if we take a program initialised by $x \in [-1,1]$ and $y=0$, and at each iteration, we rotate the point by $45^\circ$, the least invariant is an 8-point star, and its best approximation inside the abstract domain of intervals is the square $[-1,1]^2$. However, this square is not an inductive invariant: no rectangle (product of intervals) is stable under the iterations, thus there is no abstract inductive invariant within the interval domain.
Using the domain of convex polyhedra, one would obtain a regular octagon.}
\label{fig:unstable}
\end{figure}

\subsubsection{Stability Inequalities}
Consider a program fragment: \lstinline@while (c) { p; }@.
We have domain definition formulas $F \definedAs L_1(s_1, \dots, s_m) \leq p_1 \land \dots \land L_n(s_1, \dots, s_m) \leq p_n$ for the precondition of the program fragment , and $F' \definedAs L'_1(s_1, \dots, s_m) \leq p'_1 \land \dots \land L'_n(s_1, \dots, s_m) \leq p'_{n'}$ for the invariant.

Define $G = \sem{c} \land \sem{p}$. $G$ is a formula whose free variables are $s_1, \dots, s_m, s'_1, \dots, s'_m$ such that  $(s_1, \dots, s_m, s'_1, \dots, s'_m) \models G$ if and only if the state $(s'_1, \dots, s'_m)$ can be reached from the state $(s_1, \dots, s_m)$ in exactly one iteration of the loop. A set $W \subseteq \bbQ^m$ is said to be an \emph{inductive invariant} for the head of the loop if $\forall \vec{s} \in W, \forall \vec{s'}~(\vec{s}, \vec{s'}) \models G \implies \vec{s'} \in W$. We seek inductive invariants of the shape defined by $F'$, thus solutions for $\vec{p'}$ of the \emph{stability condition}:
\begin{equation}\label{eqn:stability}
\forall \vec{s}, \vec{s'}~ F' \land G \implies F'[\vec{s'}/\vec{s}].
\end{equation}

Not only do we want an inductive invariant, but we also want the initial states of the program to be included in it. The condition then becomes
\begin{equation}\label{eqn:inductive_invariant}
H \definedAs (\forall \vec{s}, F \implies F') \land
(\forall \vec{s}, \vec{s'}~ F' \land G \implies F'[\vec{s'}/\vec{s}])
\end{equation}
This is an invariant condition for the head $A$ of a loop written as:
\begin{lstlisting}
loop:
  /* A */
  if (! c) goto end;
  /* B */
  loop body
  goto loop;
end:
\end{lstlisting}

Alternatively, one can consider an invariant condition for location~$B$. The condition then becomes
\begin{equation}\label{eqn:inductive_invariant_alt}
H_{\textrm{alt}} \definedAs \forall \vec{s}~
  \sem{c} \land
  (F \lor (\exists \vec{s''}~  \sem{p}[\vec{s''}/\vec{s},\vec{s}/\vec{s'}] \land F'[\vec{s''}/\vec{s}])) \implies F'
\end{equation}
This alternate condition is very similar to the previous one (a universally quantified formula with no alternation, since the $\exists$ is negated). For the sake of simplicity, we shall only discuss the treatment of formula~$H$; formula $H_{\textrm{alt}}$ can be treated in the same way.

This formula links the values of the input constraint parameters $p_1, \dots, p_n$ to acceptable values of the invariant constraint parameters $p'_1, \dots, p'_{n'}$.
In the same way that our soundness or correctness condition on abstract transformers allowed any sound post-condition, whether optimal or not, this formula allows any inductive invariant of the required shape as long as it contains the precondition, not just the least one.

The intersection of sets defined by $\vec{p'}_1$ and $\vec{p'}_2$ is defined by $\min(\vec{p'}_1, \vec{p'}_2)$. More generally, the intersection of a family of sets, unbounded yet closed under intersection, defined by $\vec{p'} \in Z$ is defined by $\min \{ p' \mid p' \in Z\}$.
We take for $Z$ the set of acceptable parameters $\vec{p'}$ such that $\vec{p'}$ defines an inductive invariant and $\forall \vec{s}, F \implies F'$; that is, we consider only inductive invariants that contain the set $I = \{ \vec{s} \mid F \}$ of precondition states.

We deduce that $p'_i$ is uniquely defined by:
$p'_i = \min(\exists p'_1, \allowbreak\dots,\allowbreak p'_{i-1},\allowbreak p'_{i+1},\allowbreak \dots,\allowbreak p'_{n'}~H)$
which can be rewritten as
\begin{equation}\label{eqn:optimality}
O'_i \definedAs (\exists p'_1, \dots, p'_{i-1}, p'_{i+1}, \dots, p'_{n'}~H) \land
(\forall \vec{q'}~H[\vec{q'}/\vec{p'}] \implies p'_i \leq q'_i)
\end{equation}
The free variables of this formula are $p_1, \dots, p_n, p'_i$. This formula defines a function (in the mathematical sense) defining $p'_i$ from $p_1, \dots, p_n$. As before, this function can be compiled to an executable version using cascades or trees of tests.

\begin{lemma}\label{lem:loop_correctness}
The formula $O'_i$ defined at Eq.~\ref{eqn:optimality} defines the least value of $p'_i$ for an inductive invariant of the shape defined by~$F$ for the transition relation defined by~$G$.
\end{lemma}

\begin{proof}
Similarly as for lemma~\ref{lem:loopless_correctness},
the formula $H$ defined at Eq.~\ref{eqn:inductive_invariant} defines a set $Y$ of admissible $p'_1, \dots, p'_{n'}$ such that $F \definedAs L_1(s_1, \dots, s_m) \leq p'_1 \land \dots \land L_n(s_1, \dots, s_m) \leq p'_{n'}$ is an inductive invariant for the loop.
Formula $O'_i$ defines $p'_i$ to be $\inf \{ p''_i \mid (p''_1, \dots, p''_n) \in Y \}$. in other words, $(p'_1, \dots, p'_{n'}) = \inf Y$.
\end{proof}
\bigskip

Thus the overall operation of the analysis method:

We start from quantified formulas $O'_i$ defining the least inductive invariant of the loop \emph{in the abstract domain} (Lemma~\ref{lem:loop_correctness}). We eliminate quantifiers from these formulas; since this does not change their models, the resulting formulas without quantifiers also define the least inductive invariant in the abstract domain.
\begin{itemize}
\item Either the problem had no precondition parameters $p_1, \dots, p_n$, and
  thus each formula $O'_i$ has only one variable $p'_i$. It consists of linear
  (in)equalities, and has a single model for $p'_i$,
  which is straightforward to extract.
  Collect these values, one obtains the values defining the least 
  invariant $(p'_1,\dots,p'_{n'})$ in the abstract domain.
\item Either the problem has precondition parameters and one employs one of the
  methods in \S\ref{part:template_linear_constraints_gen} to obtain equivalent
  executable code.
\end{itemize}

The overall correctness of the method is quite tautological. We start from a non-executable specification of the least inductive invariant in the abstract domain in the form of quantified formulas. We eliminate the quantifiers from these formulas, then process them into equivalent executable code. At all steps, we have preserved logical equivalence with the original definition. In short, we have synthetized the implementation of the best transformer from its specification.

\subsubsection{Simple Loop Example}
\label{part:loop_example}
To show how the method operates in practice, let us consider first a very simple example (\lstinline|nondet()| is a nondeterministic choice):
\begin{lstlisting}
int i=0;
while (i <= n) {
  if (nondet()) {
    i=i+1;
    if (i == n) {
      i=0;
    }
  }
}
\end{lstlisting}

Let us abstract \texttt{i} at the head of the loop using an interval $[i_{\min},i_{\max}]$. For simplicity, we consider the case where the loop is at least entered once, and thus $i=0$ belongs to the invariant.
For better precision, we model each comparison $x \neq y$ over the integers as
$x >= y+1 \lor x <= y-1$; similar transformations apply for other operators.
The formula expressing that such an interval is an inductive invariant is:%
\begin{multline}
i_{\min}\leq 0\land 0\leq
   i_{\max}\land \forall i \forall i'~
   ((i_{\min}\leq i\land i\leq i_{\max}\land\\
   (((i+1\leq n-1\lor i+1\geq
   n+1)\land i'=i+1)\lor\\
   (i+1=n+1\land i'=0)\lor
   i'=i))\implies
   (i_{\min}\leq i'\land
   i'\leq i_{\max}))
\end{multline}

Quantifier elimination produces:
\begin{multline}
(i_{\min}\leq 0\land
   i_{\max}\geq 0\land
   i_{\max}<n\land
   -i_{\min}+n-2<0)\lor\\
   (i_{\min}\leq 0\land
   i_{\max}\geq 0\land
   i_{\max}-n+1\geq 0\land
   i_{\max}<n)
\end{multline}

The formulas defining optimal $i_{\min}$ and $i_{\max}$ are:
\begin{eqnarray}
i_{\min}\geq 0\land
   i_{\min}\leq 0\land n>0\\
(i_{\max}= 0\land n>0\land n<2) \lor
(i_{\max} = n-1 \land i_{\max}\geq 1)
\end{eqnarray}

We note that this invariant is only valid for $n > 0$, which is unsurprising given that we specifically looked for invariants containing the precondition $i = 0$. The output abstract transfer function is therefore:
\begin{lstlisting}
if (n <= 0) {
  fail();
} else {
  iMin = 0;
  if (n < 2) {
    iMax = 0;
  } else /* n >= 2 */
    iMax = n-1;
  }
}
\end{lstlisting}
The case disjunction \texttt{n < 2} looks unnecessary, but is a side effect of the use of rational numbers to model a problem over the integers. The resulting abstract transfer function is optimal, but on such a simple case, one could have obtained the same using polyhedra \cite{CousotHalbwachs78} or octagons \cite{Mine_AST_WCRE01}.

We have already noted (\S\ref{part:absint}) that even with we replace \lstinline|n| by the constant~10, the classical widening/narrowing approach will fail to identify the least invariant of this loop, and that extra techniques such have widening with thresholds have to be used.

\subsubsection{Synchronous Data Flow Example: Rate Limiter}
\label{part:rate_limiter}
To go back to the original problem of floating-point data in data-flow languages, let us consider the following library block: a \emph{rate limiter}. As seen in Listing~\ref{ex:rate_limiter}, such a block in inserted in a reactive loop, as shown below, where \lstinline@assume(c)@ stands for \lstinline@if (c) {} else {fail();}@ and \lstinline|fail()| aborts execution.
\begin{lstlisting}[float,caption={Rate limiter. },label=ex:rate_limiter]
while (true) {
  ...
  e1 = random(); assume(e1 >= e1min && e1 <= e1max);
  e2 = random(); assume(e2 >= e2min && e2 <= e2max);
  e3 = random(); assume(e3 >= e3min && e3 <= e3max);
  olds1 = s1;
  if (nondet()) {
    s1 = e3;
  } else {
    if (e1 - olds1 < -e2) {
      s1 = olds1 - e2;
    }
    if (e1 - olds1 > e2) {
      s1 = olds1 + e2;
    }
  }
  ...
}
\end{lstlisting}

This block has three input streams \lstinline|e1|, \lstinline|e2|, and \lstinline|e3|, and one output stream \lstinline|s1|. In intuitive terms, \lstinline|s1| is the same as \lstinline|e1| but where the maximal slope of the signal between two successive clock ticks is bounded by \lstinline|e2|, thus the name \emph{rate limiter}. At some points in time (modelled by nondeterministic choice), the value of the signal is reset to that of the third input~\lstinline|e3|.

We are interested in the input-output behaviour of that block: obtain bounds on the output \lstinline|s1| of the system as functions of bounds on the inputs (\lstinline|e1|, \lstinline|e2|, \lstinline|e3|).
One difficulty is that the \texttt{s1} output is memorised, so as to be used as an input to the next computation step. The semantics of such a block is therefore expressed as a fixed point.

We wish to know the least inductive invariant of the form
${s_1}_{\textrm{min}} \leq s_1 \leq {s_1}_{\textrm{max}}$ under the assumption that ${e_1}_{\textrm{min}} \leq {e_1}_{\textrm{max}} \land
{e_2}_{\textrm{min}} \leq {e_2}_{\textrm{max}} \land
{e_3}_{\textrm{min}} \leq {e_3}_{\textrm{max}}$.
The stability condition yields, after quantifier elimination and projection on ${s_1}_{\textrm{max}}$ of the condition ${s_1}_{\textrm{max}} \geq {e_1}_{\textrm{max}} \land {s_1}_{\textrm{max}} \geq {e_3}_{\textrm{max}}$. Minimisation then yields an expression that can be compiled to an if-then-else tree (Listing~\ref{ex:itetree}).

\begin{lstlisting}[float,caption={If then else tree},label=ex:itetree]
if (e1max > e3max) {
  s1max = e1max;
} else {
  s1max = e3max;
} 
\end{lstlisting}

This result, automatically obtained, coincides with the intuition that a rate limiter (at least, one implemented with exact arithmetic) should not change the range of the signal that it processes.

This program fragment has a rather more complex behaviour if all variables and operations are IEEE-754 floating-point, since rounding errors introduce slight differences of regimes between ranges of inputs. Rounding errors in the program to be analysed introduce difficulties for analyses using widenings, since invariant candidates are likely to be ``almost stable'', but not truly stable, because of these errors. Again, there exist workarounds so that widening-based approaches can still operate \cite[Sec.~7.1.4]{ASTREE_PLDI03}. We shall see in \S\ref{part:float} how to correctly abstract floating-point behaviour within our framework; unfortunately, the formulas produced tend to be very large due to many case disjunctions. The implementation of the abstract transformer produced for the above rate limiter in floating-point does not fit within one page of article, this is why we omitted it.

\section{Extensions of the framework using real linear arithmetic}
\label{part:extensions}
We shall describe here a few extensions to the class of programs and domains that we can handle, all of which are still based on quantifier elimination over real linear arithmetic. (In \S\ref{part:other_numerical_domains}, we shall investigate extensions using other arithmetic theories.)

\subsection{Emptiness}
We have so far supposed that the statement where the inductive invariant is computed is reachable, and thus that there exists some nonempty set of initial states that constrain the inductive invariant from below. More generally, and especially for the constructs described in \S\ref{part:partitioning} and \S\ref{part:noop_nests}, it may be necessary to encode the bottom element $\bot$ of the abstract domain, which represents the empty set of states.
This can be done using one Boolean variable per element that might be empty: instead of template parameters $(p_1,\dots,p_n)$, we will have $(b,p_1,\dots,p_n)$, with the semantics that $\gamma(\false,p_1,\dots,p_n) = \emptyset$ and $\gamma(\true,p_1,\dots,p_n) = \{ \vec{s} \mid \forall i~ L_i(\vec{s}) \leq p_i \}$.

\citet{Sankaranarayana+others/05/Scalable} use $p_i = -\infty$ for this, but in our framework, infinities themselves have to be encoded using Booleans, as we'll see in the next subsection. Furthermore, if we have an abstract element in normal form with a constraint $L_i(v_1, \dots) \leq -\infty$, it means that all $p_j$ are $-\infty$ and thus it is sufficient to have a single Boolean variable for all of them.

\subsection{Infinities}
\label{part:infinities}
The techniques explained in Sec.~\ref{part:template_linear_constraints} allow only finite bounds. Consider for instance the following program:
\begin{lstlisting}
x = 0;
while (true) {
  x = x+1;
}
\end{lstlisting}
We would like to obtain as a result of the analysis that $x$ lies in $[0,\infty)$. Yet, what will happen is that the formula describing the couples $(x_{\min},x_{\max})$ defining inductive invariants $x_{\min} \leq x \leq x_{\max}$ will be simplified to $\false$.

More annoyingly, with the following program:
\begin{lstlisting}
x = y;
while (true) {
  x = x+1;
}
\end{lstlisting}
we would at least hope to infer that $x_{\min} = y_{\min}$. Yet, if we look for an invariant of the form $x_{\min} \leq x \leq x_{\max}$, there is no solution for any value of $(y_{\min},y_{\max})$! In \S\ref{part:presburger_abstraction} we shall see an example where it is actually interesting to infer the domain of values of the precondition where the least inductive invariant interval is finite, and that this domain can simply be obtained by existential quantification on the parameters of the inductive invariant followed by elimination. But in general, this is not what we want; instead, we would prefer to allow infinite values in the $p$ and~$p'$.

This can be easily achieved by a minor alteration to our definitions. Each parameter $p_i$ is replaced by two parameters $p^b_i$ and $p^\infty_i$. $p^\infty_i$ is constrained to be in $\{0,1\}$ (if the quantifier elimination procedure in use allows Boolean variables, then $p^\infty_i$ can be taken as a Boolean variable); $p^\infty_i=0$ means that $p_i$ is finite and equal to $p_i^b$, $p^\infty_i=1$ means $p_i = +\infty$. $L_i \leq p_i$ becomes $(p^\infty_i > 0) \lor (L_i \leq p^b_i)$, $L_i < p_i$ becomes $(p^\infty_i > 0) \lor (L_i < p^b_i)$. After this rewriting, all formulas are formulas of the theory of linear inequalities without infinities and are amenable to the appropriate algorithms.

Unfortunately, the added combinatorial complexity induced by these Boolean variables tends to lead to intolerable computation times in the quantifier elimination procedures. Further work is needed, probably in the direction of better quantifier elimination procedures for combinations of Boolean and real quantified variables. Alternatively, one can envision directly including reasoning about infinities inside these procedures, though this is of course a delicate matter because of the possibility of generation of indeterminate forms~$\infty-\infty$ if formulas are handled without special care.

\subsection{Non-Convex Domains}
\label{part:non-convex}
Section~\ref{part:template_linear_constraints} constrains formulas to be conjunctions of inequalities of the form $L_i \leq p_i$. What happens if we consider formulas that may contain disjunctions?

The template linear constraint domains of section~\ref{part:template_linear_constraints} have a very important property: they are closed under (infinite) intersection; that is, if we have a family $\vec{p} \in W$, then there exists $p_0$ such that $\bigcap_{\vec{p} \in W} \gamma_F(\vec{p}) = \gamma_F(\vec{p}_0)$ (besides, $p_0 = \inf \{ \vec{p} \mid \vec{p} \in W \}$). This is what enables us to request the \emph{least} element that contains the exact post-condition, or the least inductive invariant in the domain: we take the intersection of all acceptable elements.

Yet, if we allow non-convex domains, there does not necessarily exist a least element $\gamma_F(\vec{p})$ such that $S \subseteq \gamma_F(\vec{p})$. Consider for instance $S = \{0, 1, 2\}$ and $F$ representing unions of two intervals $((-x \leq p_1 \land x \leq p_2) \lor (-x \leq p_3 \land x \leq p_4)) \land p_2 \leq -p_3$. There are two, incomparable, minimal elements of the form $\gamma_F(\vec{p})$: $p_1=p_2=0 \land p_3=-1 \land p_4=2$ and $p_1=0 \land p_2=1 \land p_3=-2 \land p_4=2$.

We consider formulas $F$ built out of linear inequalities $L_i(s_1, \allowbreak\dots,\allowbreak s_n) \leq p_i$ as atoms, conjunctions, and disjunctions. 
By induction on the structure of $F$, we can show that $\gamma_F: (\bbR \cup \{ -\infty \})^n \rightarrow \parts{\bbR^n}$ is inf-continuous; that is, for any descending chain $(\vec{p}_i)_{i \in I}$ such that $\lim_i \vec{p}_i = \vec{p}_\infty$, then $\gamma_F(\vec{p}_i)$ is decreasing and $\bigcap_{i \in I} \gamma_F(\vec{p}_i) = \gamma_F(\vec{p}_\infty)$. The property is trivial for atomic formulas, and is conserved by greatest lower bounds ($\land$) as well as binary least upper bounds ($\lor$).

Let us consider a set $S \subseteq \parts{\bbR^n}$, stable under arbitrary intersection (or at least, greatest lower bounds of descending chains). $S$ can be for instance the set of invariants of a relation, or the set of over-approximations of a set~$W$. $\gamma_F^{-1}(S)$ is the set of suitable domain parameters; for instance, it is the set of parameters representing inductive invariants of the shape specified by $F$, or the set of representable over-approximations of~$W$. $\gamma_F^{-1}(S)$ is stable under greatest lower bounds of descending chains: take a descending chain $(\vec{p}_i)_{i \in I}$, then $\gamma_F(\lim_i \vec{p}_i) = \cap_i \gamma_F(\vec{p}_i) \in S$ by inf-continuity and stability of $S$. By Zorn's lemma, $\gamma_F^{-1}(S)$ has at least one minimal element.

Let $P[\vec{p}]$ be a formula representing~$\gamma_F^{-1}(S)$ (Sec.~\ref{part:template_linear_constraints} proposes formulas defining safe post-conditions and inductive invariants). The formula $G[\vec{p}] \definedAs P[\vec{p}] \land \forall \vec{p'}~ P[\vec{p'}] \land \vec{p'} \leq \vec{p} \implies \vec{p} \leq \vec{p'}$ defines the minimal elements of $\gamma^{-1}(S)$.

For instance, consider $\vec{p}=(a,b,c,d)$, $F \definedAs (-x \leq a \land x \leq b) \lor (-x \leq c \land x \leq d)$, representing unions of two intervals $[-a, b] \cup [-c,d]$. We want upper-approximations of the set $\{0, 1, 3\}$; that is $P[\vec{p}] \definedAs \forall x~ (x=0 \lor x=1 \lor x=3 \implies F[\vec{p},x])$. We add the constraint that $-a \leq b \land b \leq -c \land -c \leq d$, so as not to obtain the same solutions twice (by exchange of $(a,b)$ and $(c,d)$) or solutions with empty intervals. By quantifier elimination over $G$, we obtain
$(a = 0 \land b = 1 \land c = -3 \land d=3) \lor (a = 0 \land b = 0 \land c = -1 \land d=3)$, that is, either $[0,0] \cup [1, 3]$ or $[0,1] \cup [3, 3]$.

\subsection{Domain Partitioning}
\label{part:partitioning}
\begin{figure}
\begin{center}
\includegraphics{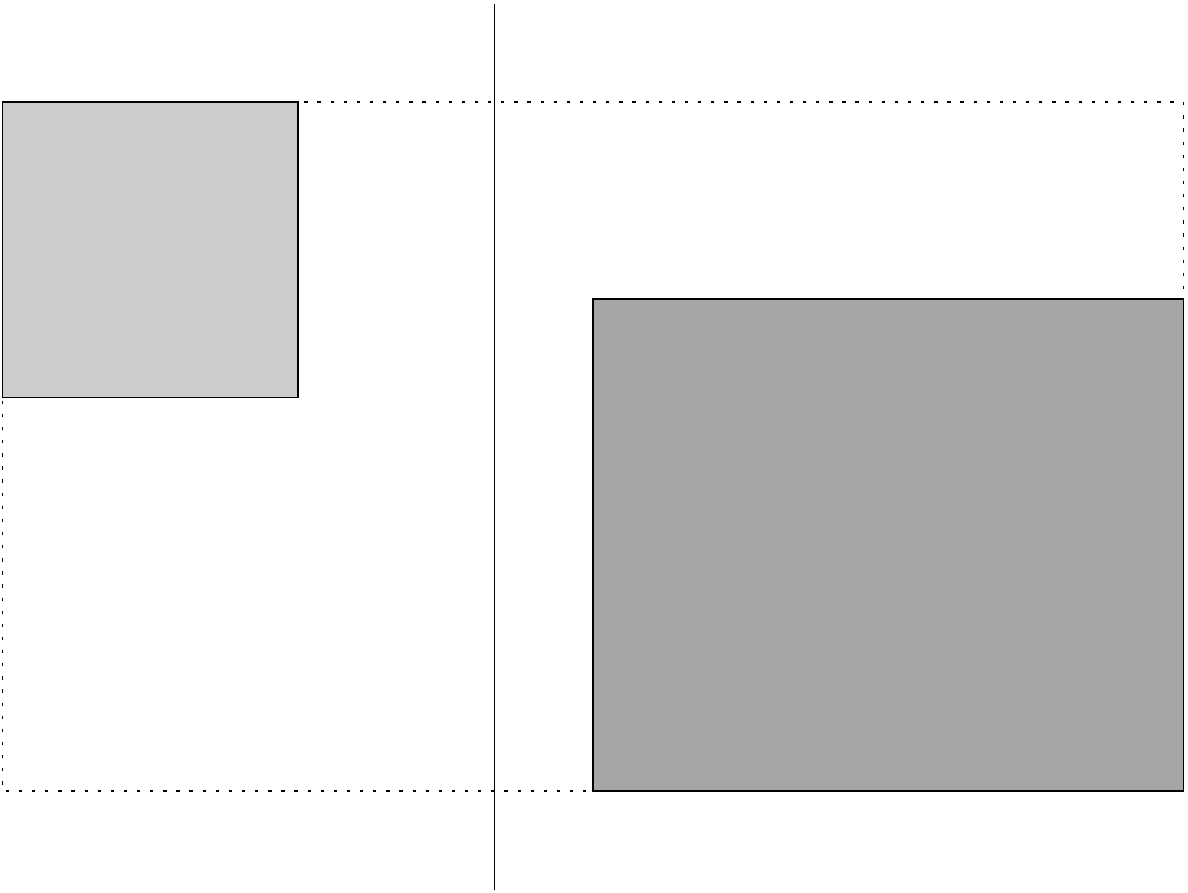}
\end{center}

\caption{The state space is partitioned into $x < 0$ and $x \geq 0$, and on each element of the partition we have a product of intervals, respectively $[-5,-2] \times [4, 7]$ and $[1,7] \times [0,5]$. Without the disjunction, we would have had to consider the much larger set $[-5,7] \times [0,7]$.}
\label{fig:partition01}
\end{figure}

Non-convex domains, in general, are not stable under intersections and thus ``best abstraction'' problems admit multiple solutions as minimal elements of the set of correct abstractions. As explained in the preceding subsection, is not very satisfactory nor efficient for analysis. There are, however, non-convex abstract domains that are stable under intersection and thus admit least elements as well as the template linear constraint domains of Sec.~\ref{part:template_linear_constraints}: those defined by partitioning of the state space.

If, for instance, we know that a program or system behaves differently according to the sign of $x$, then we can decide \emph{in advance} to partition the state space into $x \geq 0$ and $x < 0$. On each element of the partition, we can have a separate template (example Fig.~\ref{fig:partition01}). We can accommodate this within our framework.

More formally, consider pairwise disjoint subsets $(C_i)_{i \in I}$ of the state space $\bbQ^m$, and abstract domains stable under intersection $(S_i)_{i \in I}$, $S_i \subseteq \parts{C_i}$. Elements of the partitioned abstract domain are unions $\bigcup_{i \in I} s_i$ where $s_i \in S_i$. If $\left(\bigcup_i s_{i,j}]\right)_{j \in J}$ is a family of elements of the domain, then $\bigcap_{j \in J}\left(\bigcup_{i \in I} s_{i,j}\right) = \bigcup_{i \in I} \bigcap_{j \in J} s_{i,j}$; that is, intersections are taken separately in each~$C_i$. in other words, the parameters of the templates on each element of the partition can be dealt with independently of each other.

Note the difference with the general disjunctive domains of \S\ref{part:non-convex}: in the general disjunctive domains, there is no partition fixed \emph{a priori}, this is why we may have several incomparable minimal elements $[0,0] \cup [1,2]$ and $[0,1] \cup [2,2]$ in the domain of disjunctions of two intervals representing the same set $\{0,1,2\}$. For a fixed partition $C_i$ and corresponding domains $S_i$, for any set $X$, for every $i$, there is a least element representing $X \cap C_i$ in domain~$S_i$. This motivates the following construct:

Take a family $(F_i[\vec{p}])_{i \in I}$ of formulas defining template linear constraint domains (conjunctions of linear inequalities $L_i(s_1,\allowbreak \dots,\allowbreak s_n) \leq p_i$) and a family $(C_i)_{i \in I}$ of formulas such that for all $i$ and $i'$, $C_i \land C_{i'}$ is equivalent to \textsf{false} and $C_1 \lor \dots \lor C_l$ is equivalent to \textsf{true}. $F = (C_1 \land F_1) \lor \dots \lor (C_l \land F_l)$ then defines an abstract domain such that $\gamma_F$ is a inf-morphism. All the techniques of \S\ref{part:template_linear_constraints} then apply.

For instance, by choosing $C_1 \definedAs x \geq 0$, $C_2 \definedAs x < 0$, $F_1 \definedAs x_{\min}^1 \leq x \leq x_{\max}^1 \land y_{\min}^1 \leq x \leq y_{\max}^1$, and $F_2 \definedAs x_{\min}^2 \leq x \leq x_{\max}^2 \land y_{\min}^2 \leq x \leq y_{\max}^2$, we can obtain Figure~\ref{fig:partition01}.

The above constructions are equivalent to assigning a separate control point to each element in the partition, with guards leading to these points according to the $C_i$, and then performing as described in~\S\ref{part:noop_nests}.

\subsection{Floating-Point Computations}
\label{part:float}
Real-life programs do not operate on real numbers; they operate on fixed-point or floating-point numbers. Floating point operations have few of the good algebraic properties of real operations; yet, they constitute approximations of these real operations, and the \emph{rounding error} introduced can be bounded.

In IEEE floating-point \cite{IEEE-754}, each atomic operation (noting $\oplus$, $\ominus$, $\otimes$, $\oslash$, $\sqrt{}_f$ for operations so as to distinguish them from the operations $+$, $-$, $\times$, $/$, $\sqrt{}$ over the reals) is mathematically defined as the image of the exact operation over the reals by a rounding function.%
\footnote{We leave aside the peculiarities of some implementations, such as those of most C compilers over the 32-bit Intel platform where there are ``extended precisions'' types used for some temporary variables and expressions can undergo double rounding~\cite{Monniaux_TOPLAS08}.}
This rounding function, depending on user choice, maps each real $x$ to the nearest floating-point value $r_n(x)$ (\emph{round to nearest mode}, with some resolution mechanism for non representable values exactly in the middle of two floating-point values), $r_{-\infty}(x)$ the greatest floating-point value less or equal to $x$ (\emph{round toward $-\infty$}), $r_{+\infty}(x)$ the least floating-point value greater or equal to $x$ (\emph{round toward $+\infty$}), $r_0(x)$ the floating-point value of the same sign as $x$ but whose magnitude is the greatest floating-point value less or equal to $|x|$ (\emph{round toward $0$}). If $x$ is too large to be representable, $r(x)=\pm\infty$ depending on the size of $x$

The semantics of the rounding operation cannot be exactly represented inside the theory of linear inequalities.%
\footnote{To be pedantic, since IEEE floating-point formats are of a finite size, the rounding operation could be exactly represented by enumeration of all possible cases; this would anyway be impossible in practice due to the enormous size of such an enumeration.}
As a consequence, we are forced to use an axiomatic over-approximation of that semantics: a formula linking a real number $x$ to its rounded version~$r(x)$.

\citet{Mine_ESOP04} uses an inequality $|r(x)-x| \leq \erel\cdot |x| + \eabs$, where $\erel$ is a \emph{relative error} and $\eabs$ is an \emph{absolute error}, leaving aside the problem of overflows. The relative error is due to rounding at the last binary digit of the significand, while the \emph{absolute error} is due to the fact that the range of exponents is finite and thus that there exists a least positive floating-point number and some nonzero values get rounded to zero instead of incurring a relative error (or get rounded to a denormal, see below).

Because our language for axioms is richer than the interval linear forms used by Min\'e, we can express more precise properties of floating-point rounding. We recall briefly the characteristics of IEEE-754 floating-point numbers. 
Nonzero floating point numbers are represented as follows:
$x = \pm 2^e m$ where $1 \leq m < 2$ is the \emph{mantissa} or
\emph{significand}, which
has a fixed number $p$ of bits, and $e$ ($E_{\min} \leq e \leq E_{\max}$
is the \emph{exponent}).
The difference introduced by changing the last binary digit of the
mantissa is $\pm s.\elast$ where $\elast = 2^{-(p-1)}$:
the \emph{unit in the last place} or \emph{ulp}.
Such a decomposition is unique for a given number
if we impose that the leftmost digit of the
mantissa is $1$ --- this is called a \emph{normalised representation}.
Except in the case of numbers of very small magnitude, IEEE-754 always
works with normalised representations. There exists a least positive normalised number $m_{\textrm{normal}}$ and a least positive denormalized number $m_{\textrm{denormal}}$, and the denormals are the multiples of $m_{\textrm{denormal}}$ less than $m_{\textrm{normal}}$. All representable numbers are multiples of~$m_{\textrm{denormal}}$.
\medskip

We shall now attempt to define an imprecise axiomatisation of the relationship between the result of a floating-point operation and the result of the corresponding ideal operation over real numbers. For this, we shall distinguish operations plus and minus on the one hand, multiplication and division on the other hand, since the former have properties of which we can take advantage.

Let us consider addition or subtraction $x=\pm a \pm b$. Suppose that $0 \leq x \leq m_{\textrm{normal}}$. $a$ and $b$ are multiples of $m_{\textrm{denormal}}$ and thus $a-b$ is exactly represented as a denormalized number; therefore $r(x)=x$. If $x > m_{\textrm{normal}}$, then $|r(x)-x| \leq \erel.x$. In other words, if the result of an addition or subtraction is denormal, then it is exact.%
\footnote{William Kahan has written extensively about the advantages of the existence and proper handling of denormals, otherwise known as \emph{gradual underflow}, as opposed to flushing to zero all numbers too small to be approximated by normal numbers, a process known as \emph{flush to zero}. For instance, with gradual underflow, $a \ominus b = 0$ is equivalent to $a = b$, quite a desirable property. See e.g.~\citep[p.~6]{Kahan_ARITH_17U}. Unfortunately, certain architectures do not implement gradual underflow, for the sake of efficient; then one has to use the $\textit{Round}$ and not the $\textit{Round}^\pm$ predicate.}
The cases for $x \leq 0$ are symmetrical.

We therefore obtain the following axiomatisation of the rounding of a positive real number $x$, result of a floating-point addition or subtraction, to a floating-point value $r$, using round-to-nearest:
\begin{equation}
\textit{Round}^\pm_+(r,x) \definedAs (x \leq m_{\textrm{normal}} \land r = x)
  \lor (x > m_{\textrm{normal}} \land -\erel.x \leq r - x \leq \erel.x)
\end{equation}
We then use this predicate to construct an axiomatisation for rounding of numbers of any sign:
\begin{multline}
\textit{Round}^\pm(r,x) \definedAs (x =0 \land r =0) \lor
  (x > 0 \land \textit{Round}^\pm_+(r,x)) \lor\\
  (x < 0 \land \textit{Round}^\pm_+(-r,-x))
\end{multline}
Equivalently, this last formula can be simplified into the equivalent:
\begin{multline}
\textit{Round}^\pm(r,x) \definedAs
  (-m_{\textrm{normal}} \leq x \leq m_{\textrm{normal}} \land x = r) \\
  \lor
\left(x > m_{\textrm{normal}} \land x (1-\erel) \leq r \leq x (1+\erel)\right)\\
  \lor
\left(x <-m_{\textrm{normal}} \land x (1+\erel) \leq r \leq x (1-\erel)\right)
\end{multline}

Consider now multiplication or division $x = a \otimes b$ or $x = a \oslash b$. Here, we cannot assume that if the result is denormal, then it is exact. A correct axiomatization of rounding for positive numbers is:
\begin{multline}
\textit{Round}_+(r,x) \definedAs
  (x \leq m_{\textrm{normal}} \land r \geq 0 \land x-\eabs \leq r \leq x+\eabs)\\
  \lor (x > m_{\textrm{normal}} \land -\erel.x \leq r - x \leq \erel.x)
\end{multline}
where $\eabs = m_{\textrm{denormal}}/2$. Now for rounding for any sign:
\begin{multline}
\textit{Round}(r,x) \definedAs (x =0 \land r =0) \lor
  (x > 0 \land \textit{Round}_+(r,x)) \lor\\
  (x < 0 \land \textit{Round}_+(-r,-x))
\end{multline}
\medskip

To each floating-point expression $e$, we associate a ``rounded-off'' variable $r_e$, the value of which we constrain using $\textit{Round}^\pm(r_e,e)$ or $\textit{Round}(r_e,e)$. For instance, a expression $e = a \oplus b$ is replaced by a variable $r_e$, and the constraint $\textit{Round}^\pm(r_e, a+b)$ is added to the semantics. In the case of a compound expression $e = ab+c$, we introduce $e_1 = ab$, and we obtain $\textit{Round}^\pm(r_e, r_{e_1}+c) \land \textit{Round}(r_{e_1}, ab)$. If we know that the compiler uses a fused multiply-add operator, we can use $\textit{Round}(r_e,ab+c)$ instead.

The drawbacks of such axiomatization of floating-point operations are that they introduce case disjunctions into formulas, leading to extra work for quantifier elimination procedures, and also that they make very unnatural coefficients appear --- that is, rational numbers with large numerator and denominator, such as $1+\erel$ or $1-\erel$, or very large integers such as $1/m_{\textrm{normal}}$. To go back to our very simple rate limiter running example (\S\ref{part:rate_limiter}), analysis times are multiplied by 12 if one stops assuming that floating-point variables behave like reals, and instead uses some axiomatization~\citep[p.~9]{Monniaux_LPAR08}. Furthermore, the formula produced is very large and hardly readable, doing many case disjunctions. Perhaps it is possible to simplify such large formulas by allowing some limited overapproximation --- for instance, we can replace the function mapping $p$ to $p'$ as defined by $(0 \leq p \leq 1 \leq p' = p) \lor (p > 1 \leq p' = (1+\epsilon)p)$ by the function defined by $0 \leq p \land p' = (1+\epsilon)p$, since the latter formula always gives a solution for $p'$ greater or equal to that of the former. Even better, such simplifications could be performed during the elimination procedure. Further investigations are needed in that respect.

\subsection{Integers}
\label{part:integers}
We have mentioned in \S\ref{part:qf_lia} that Presburger arithmetic admits quantifier elimination. We therefore could apply quantifier elimination in Presburger arithmetic, similarly as we do with respect to real linear arithmetic. Unfortunately, as we shall see in \S\ref{part:presburger_abstraction}, such an approach suffers from explosion in the size of the formulas and the cost of the algorithms.

Instead, we used a \emph{relaxation} approach: all integers are treated as reals; strict inequalities $a < b$ where both sides are integers are recoded $a \leq b-1$. For instance, if the program contains a if-then-else over $x \leq y$, then $x \leq y$ is a precondition for the ``then'' branch, and $x \geq y+1$ is a precondition of the ``else'' branch. Note that this means that traces of execution such that $y < x < y+1$ are considered to fail.

In some cases, such as the McCarthy 91 function example from \S\ref{part:McCarthy91}, it is necessary to constraint the reasoning procedures so that they consider that the negation of $a \leq b$ is $a \geq b+1$. We hope that improvements of quantifier elimination algorithms will be able to allow a more elegant approach.

Another issue is that in many programming languages, integers are bounded and that arithmetic operations are actually performed modulo $2^n$ (with $n$ typically 8, 16, 32 or 64). The problem then lies within an enormous, but finite, state space. Clever techniques for reasoning about \emph{bit-vector arithmetic} are being investigated by the SMT-solving community. Again, we hope that future work will provide good quantifier elimination techniques for this arithmetic, or combinations thereof with the linear theory of reals.

\subsection{Nonlinear constructs}
In \S\ref{sec:real_polynomial_abstraction} we explain how to fully deal with nonlinear program constructs and templates, at the expense of very high computational complexity.

A more practical approach is to linearise the program expressions~\citep{mine:vmcai06}\cite[ch.~6]{Mine_PhD}. If we encounter an assignment $z := x*y$ and we know, perhaps through rough interval analysis, an interval $y \in [y_{\min},y_{\max}]$, we can use the following abstraction of the semantics of this assignment:
\begin{equation}\label{eqn:linearize}
(x \geq 0 \land x y_{\min} \leq z' \leq x y_{\max}) \lor
(x < 0 \land x y_{\max} \leq z' \leq x y_{\min})
\end{equation}

If we know intervals for both $x$ and $y$, then we can apply Eq.~\ref{eqn:linearize} to both $(x,[y_{\min},y_{\max}])$ and  $(y,[x_{\min},x_{\max}])$, and take the conjunction of the resulting formulas.

\section{Complex control flow}
\label{part:control-flow}
We have so far assumed no procedure call, and at most one single loop. We shall see here how to deal with arbitrary control flow graphs and call graph structures.


\subsection{Arbitrary control graph and loop nests}
\label{part:noop_nests}
In Sec.~\ref{part:template_linear_constraints_invariants}, we have explained how to abstract a single fixed point. The method can be applied to multiple nested fixed points by replacing the inner fixed point by its abstraction. For instance, assume the rate limiter of Sec.~\ref{part:rate_limiter} is placed inside a larger loop. One may replace it by its abstraction:

\begin{lstlisting}
if (e1max > e3max) {
  s1max = e1max;
} else {
  s1max = e3max;
} 
assume(s1 <= s1max);
/* and similar for s1min */
\end{lstlisting}

Alternatively, we can extend our framework to an arbitrary control flow graph with nested loops, the semantics of which is expressed as a single fixed point.
We may use the same method as proposed by \citet[\S2]{Gulwani_PLDI08} and other authors. First, a \emph{cut set} of program locations is identified; any cycle in the control flow graph must go through at least one program point in the cut set. In widening-based fixed point approximations, one classically applies widening at each point in the cut set. A simple method for choosing a cut set is to include all targets of back edges in a depth-first traversal of the control-flow graph, starting from the start node; in the case of structured program, this amounts to choosing the head node of each loop. This is not necessarily the best choice with respect to precision, though~\cite[\S2.3]{Gulwani_PLDI08}; \citet[Sec.~3.6]{Bourdoncle_PhD} discusses methods for choosing such as cut-set.

To each point in the cut set we associate an element in the abstract domain, parameterised by a number of variables. The values of these variables for all points in the cut-set define an invariant candidate. Since paths between elements of the cut sets cannot contain a cycle, their denotational semantics can be expressed simply by an existentially quantified formula. Possible paths between each source and destination elements in the cut-set defined a stability condition (Formula~\ref{eqn:stability}). The conjunction of all these stability conditions defines acceptable inductive invariants. As above, the least inductive invariant is obtained by writing a minimisation formula (Sec.~\ref{part:template_linear_constraints_invariants}).

Let us consider the loop nest in Listing~\ref{ex:loop_nest}.
\lstset{language=Cext}
\begin{lstlisting}[float,caption={Loop nest},label=ex:loop_nest]
i=0;
while(true) { /* A */
  if (choice()) {
    j=0;
    while(j < i) { /* B */
      /* something */
      j=j+1;
    }
    i=i+1;
    if (i==20) {
      i=0;
    }
  } else {
    /* something */
  }
}  
\end{lstlisting}
We choose program points $A$ and $B$ as cut-set. At program point A, we look for an invariant of the form $I_A(i,j) \definedAs i_{\min,A} \leq i \leq i_{\max,A}$, and at program point B, for an invariant of the form $I_B(i,j) \definedAs i_{\min,B} \leq i \leq i_{\max,B} \land j_{\min} \leq j \leq j_{\max} \land \delta_{\min} \leq i-j \leq \delta_{\max}$ (a \emph{difference-bound} invariant). The (somewhat edited for brevity) stability formula is written:
\begin{multline}
\forall j~ I_A(0,j)\land \forall i \forall j~
((I_B(i,j)\land j\geq i\land (i+1\leq 19\lor\\
i+1=20\lor i+1\geq 21))\Rightarrow
\text{If}[i+1=20,I_A(0,j),I_A(i+1,j)])\land\\
\forall i \forall j~ (I_A(i,j)\Rightarrow
I_B(i,0))\land \forall i \forall j~((I_B(i,j)\land
j<i)\\
\Rightarrow I_B(i,j+1))
\end{multline}
where $\text{If}[b,e_1,e_2] = (b \land e_1) \lor (\neg b \land e_2)$.

Replacing $I_A$ and $I_B$ into this formula, then applying quantifier elimination, we obtain a formula defining all acceptable tuples $(i_{\min,A}, i_{\max,A}, i_{\min,B}, i_{\max,B}, j_{\min}, j_{\max}, \delta_{\min}, \delta_{\max})$. Optimal values are then obtained by further quantifier elimination: $i_{\min,A}=i_{\min,B}=j_{\min}=0$, $i_{\max,A}=i_{\max,B}=19$, $j_{\max}=20$, $\delta_{\min}=1$, $\delta_{\max}=19$.

The same example can be solved by replacing $20$ by another variable \texttt{n} as in Sec.~\ref{part:loop_example}.

\subsection{Procedures and Recursive Procedures}
\label{part:McCarthy91}
We have so far considered abstractions of program blocks with respect to sets of program states. A program block is considered as a transformer from a state of input program states to the corresponding set of output program states. The analysis outputs a sound and optimal (in a certain way) abstract transformer, mapping an abstract set of input states to an abstract set of output states.

Assuming there are no recursive procedures, procedure calls can be easily dealt with. We can simply inline the procedure at the point of call, as done in e.g. \textsc{Astr\'ee} \cite{BlanchetCousotEtAl02-NJ,ASTREE_PLDI03,ASTREE_ESOP05}. Because inlining the concrete procedure may lead to code blowup, we may also inline its abstraction, considered as a nondeterministic program. Consider a complex procedure \texttt{P} with input variable \texttt{x} and output variable \texttt{x}. We abstract the procedure automatically with respect to the interval domain for the postcondition ($z_{\min} \leq z \leq z_{\max}$); suppose we obtain $z_{\max}:=1000; z_{\min}:=x$ then we can replace the function call by \lstinline@z <= 1000 && z >= x@. This is a form of \emph{modular interprocedural analysis}: considering the call graph, we can abstract the leaf procedures, then those calling the leaf procedures and so on. We can also do likewise for nested loops: abstract the innermost loop, then the next innermost one, etc.
This method is however insufficient for dealing with recursive procedures.

In order to analyse recursive procedures, we need to abstract not sets of states, but sets of pairs of states, expressing the input-output relationships of procedures. In the case of recursive procedures, these relationships are the least solution of a system of equations.

To take a concrete example, let us consider McCarthy's famous ``91 function''~\cite{Manna_McCarthy_69,Manna_Pnueli_JACM70}, which, non-obviously, returns 91 for all inputs less than 101:
\begin{lstlisting}
int M(int n) {
  if (n > 100) {
    return n-10;
  } else {
    return M(M(n+11));
  }
}
\end{lstlisting}

The concrete semantics of that function is a relationship $R$ between its input $n$ and its output $r$. It is the least solution of
\begin{multline}\label{eqn:McCarthy91_stability}
R \supseteq \{ (n, r) \in \bbZ^2 \mid (n > 100 \land r=n-10) \lor\\
   (n \leq 100 \land \exists n_2 \in \bbZ
   (n+11, n_2) \in R \land (n_2, r) \in R) \}
\end{multline}

We look for a inductive invariant of the form $I \definedAs ((n \geq A) \land (r-n \geq \delta) \land (r-n \leq \Delta)) \lor ((n \leq B) \land (r=C))$, a non-convex domain (Sec.~\ref{part:non-convex}). By replacing $R$ by $I$ into inclusion~\ref{eqn:McCarthy91_stability}, and by universal quantification over $n,r,n_2$, we obtain the set of admissible parameters for invariants of this shape. By quantifier elimination, we obtain $(C=91) \land (\delta=\Delta=-10) \land (A=101) \land (B=100)$ within a fraction of a second using \textsc{Mjollnir} (see Sec.~\ref{part:experiments}).

In this case, there is a single acceptable inductive invariant of the suggested shape. In general, there may be parameters to optimise, as explained in Sec.~\ref{part:template_linear_constraints_invariants}. The result of this analysis is therefore a map from parameters defining sets of states to parameters defining sets of pairs of states (the abstraction of a transition relation). This abstract transition relation (a subset of $X \times Y$ where $X$ and $Y$ are the input and output state sets) can be transformed into an abstract transformer in $X^\sharp \rightarrow Y^\sharp$ as explained in Sec.~\ref{part:template_linear_constraints_transformers}. Such an interprocedural analysis may also be used to enhance the analysis of loops~\cite{DBLP:conf/cc/MartinAWF98}.

\section{Implementations and Experiments}

\label{part:experiments}
We have implemented the techniques of Sec.~\ref{part:template_linear_constraints_abstraction} in quantifier elimination packages, including  \textsc{Mathematica}%
\footnote{\textsc{Mathematica} is a commercial computer algebra package available under an unfree license from \href{http://www.wolfram.com}{Wolfram Research} \citep{Mathematica_book}.}
and \textsc{Reduce}~3.8\footnote{\href{http://www.reduce-algebra.com/}{\textsc{Reduce}} is a computer algebra package from Anthony C. Hearn, now available under a modified BSD licence.}
+ \textsc{Redlog}\footnote{\href{http://www.redlog.eu}{\textsc{Redlog}} is an extension to \textsc{Reduce} for working over quantified formulas.}
in addition to our own package, \textsc{Mjollnir}~\cite{Monniaux_LPAR08}.%
\footnote{\textsc{Mjollnir} is available under a free license from the \href{http://www-verimag.imag.fr/~monniaux/mjollnir.html.en}{author's home page}. In addition to the author's own quantifier elimination techniques, it implements \citeauthor{FerranteRackoff75} and \citeauthor{LoosWeispfenning93}'s.}
We ignore which exact techniques are implemented within \textsc{Mathematica}.
\footnote{\citeauthor{LoosWeispfenning93}'s quantifier elimination procedure is used by \textsc{Mathematica} to perform simplifications over linear inequalities~\cite[\S A.9.5]{Mathematica_book}, but we are unsure whether this is the algorithm called by the \texttt{Reduce} function.}
\textsc{Redlog} appears to implement some virtual substitution method \citep{Redlog,Weispfenning88}.

As test cases, we took a library of operators for synchronous programming, having streams of floating-point values as input and outputs. These operators are written in a restricted subset of~C and take as much as 20 lines. A front-end based on \textsc{CIL}~\cite{Cil} converts them into formulas, then these formulas are processed and the corresponding abstract transfer functions are pretty-printed. Since for our application, it is important to bound numerical quantities, we chose the interval domain.

Among the extensions in \S\ref{part:extensions}, we implemented those relevant to floating-point (\S\ref{part:float}) and integers (\S\ref{part:integers}), and did more manual experiments with infinities (\S\ref{part:infinities}) and recursive functions (\S\ref{part:McCarthy91}).

For instance, the rate limiter presented in Sec.~\ref{part:rate_limiter} was extracted from that library. Since this operator includes a memory (a variable whose value is retained from a call to the operator to the next one), its data-flow semantics is expressed using a fixed-point. When considered with real variables, the resulting expanded formula was approximately 1000 characters long, and with floating point variables approximately 8000 characters long. Despite the length of these formulas, they can be processed by \textsc{Mjollnir} in a matter of seconds. The result can then be saved once and for all.

Analysers such as \textsc{Astr\'ee}~\cite{BlanchetCousotEtAl02-NJ,ASTREE_PLDI03,ASTREE_ESOP05} must have special knowledge about such operators, otherwise the analysis results are too coarse (for instance, the intervals do not get stabilized at all). The \textsc{Astr\'ee} development team therefore had to provide specialized, hand-written analyses for certain operators. In contrast, all linear floating-point operators in the library were analysed within a fraction of a second using the method in the present article, assuming that floating-point values in the source code were real numbers. If one considered instead the abstraction of floating-point computations using real numbers from Sec.~\ref{part:float}, computation times did not exceed 17~seconds per operator; the formulas produced are considerably more complex than in the real case. Note that this computation is done once and for all for each operator; a static analyser can therefore cache this information for further use and need not recompute abstractions for library functions or operators unless these functions are updated.

Our analyser front-end currently cannot deal with non-numerical operations and data structures (pointers, records, and arrays). It is therefore not yet capable of directly dealing with the real control programs that e.g. \textsc{Astr\'ee} accepts, which do not consist purely of numerical operators. We plan to integrate our analysis method into a more generic analyser. Alternatively, we plan to adapt a front-end for synchronous programming languages such as~\textsc{Simulink}, a tool widely used by control/command engineers.

The correctness of the methods described in this article does not rely on any particularity of the quantifier elimination procedure used, provided one also has symbolic computation procedures for e.g. putting formulas in disjunctive normal form and simplifying them. The difference between the various quantifier elimination and simplification procedures is efficiency; experiments showed that ours was vastly more efficient than the others tested for this kind of application. For instance, our implementation of our quantifier elimination algorithm~\cite{Monniaux_LPAR08} was able to complete the analysis of the rate limiter of Sec.~\ref{part:rate_limiter}, implemented over the reals, in 1.4~s, and in 17~s with the same example over floating-point numbers, while \textsc{Redlog} took 182~s for the former and could not finish the latter, and \textsc{Mathematica} could analyse neither (out-of-memory). On other examples, our quantifier elimination procedure is faster than the other ones, or can complete eliminations that the others cannot.

\section{Extensions to other numerical domains}
\label{part:other_numerical_domains}
We have so far concerned ourselves solely with real (and possibly Boolean) variables appearing in linear arithmetic formulas. In \S\ref{part:extensions}, we have seen how to reason over certain other data types (integers, floating-point values), but again modeling them as real numbers. Yet, in \S\ref{part:qe}, we pointed out that linear real arithmetic is not the only arithmetic theory with quantifier elimination algorithms; other well-known examples include Presburger arithmetic (linear integer arithmetic) and the theory of real closed fields (that is, polynomial real arithmetic).

In this section, we report on using quantifier elimination in both these theories for computing optimal transformers or fixed points. Unfortunately, experiments have shown that the high complexity of the quantifier elimination algorithms for these theories and the lack of simplifications for the formulas they produce preclude their use in practice. The results in this section are thus mainly of theoretical interest.

\subsection{Presburger arithmetic}
\label{part:presburger_abstraction}

The approach from \S\ref{part:integers} relaxes integers to reals. It therefore yields correct results, but might lead to overapproximations that could be avoided if integers were used instead. What if we used quantifier elimination on Presburger arithmetic instead?

Consider the following simple example:
\begin{lstlisting}
int a, m;
...
i = 0;
while (i < m) {
  i = i+a;
}
\end{lstlisting}

Let us abstract the loop variable \lstinline|a| using an interval $[l,h]$. These bounds must satisfy $I \definedAs l \leq 0 \leq h$ (initialisation of the loop) and
\begin{equation}
S \definedAs \forall i~ \left(l\leq i \leq u \Rightarrow (i > m \lor l \leq i+a \leq u)\right)
\end{equation}

The condition for the existence of a finite range of values for \lstinline|i| is therefore $\exists l \exists u~I \land S$.%
\footnote{This example is motivated by the fact that for $a \neq 0$, the loop terminates if and only if $[l,u]$ is finite. It is a simplified version of a loop termination problem from Paul Feautrier and Laure Gonnord.}
Intuitively, this condition is equivalent to $m < 0 \lor a \geq 0$. Yet, the formula produced by the quantifier elimination for Presburger arithmetic from \textsc{Redlog} yields a very large formula, more than one page long, which we have therefore omitted.%
\footnote{It could be that \textsc{Redlog} gives erroneous answers. At some point, we generated a formula $F$ with free variables $a,m$ such that \textsc{Redlog} produced \false when eliminating quantifiers from $\exists m \exists a~F$, and produced \true  when eliminating quantifiers from $\exists a \exists m~F$; at another point we made \textsc{Reduce}/\textsc{Redlog} crash with a segmentation fault.}
\textsc{Redlog} cannot even perform simplifications such as replacing $i=0 \lor i=1 \lor i=2$ by $0 \leq i \leq 2$.

In comparison, the real relaxation gives exact and fast results:
\begin{equation}
S_\bbR \definedAs \forall i~ \left(i < l \lor i > u \lor i \geq m+1 \lor l \leq i+a \leq u)\right)
\end{equation}
Eliminating quantifiers from $\exists l \exists u~I \land S_\bbR$ yields immediately the answer $a \geq 0 \lor m \leq -1$.

\subsection{Real polynomial constraint domains}
\label{sec:real_polynomial_abstraction}
We now consider the abstraction of program states (in $\bbR^V$) using domains defined by polynomial constraints, a natural extension of those seen previously (\S\ref{part:template_linear_constraints}); the orthogonal extensions from \S~\ref{part:extensions} also apply. Instead of quantifier elimination in linear real arithmetic, we shall use quantifier elimination in the theory of real closed fields. One difference, though, is that we will not be able to produce nice, closed form formulas, at least not in general.

\subsubsection{Method}
We generalise the constructs of \S\ref{part:template_linear_constraints_abstraction}, except those of \S~\ref{part:template_linear_constraints_gen}, to formulas over polynomial inequalities.
The same results hold:
\begin{itemize}
\item For any loop-free program code, and any template polynomial abstract domain with parameters $p_1,\dots,p_n$, there is a family of formulas $F_1,\dots,F_n$ that uniquely defines the optimal parameters $p'_1,\dots,p'_{n'}$ of the postcondition with respect those $p_1,\dots,p_n$ in the precondition (the free variables of $F_i$ are among $p_1,\dots,p_n,p'_i$).
\item For any loop, and any template polynomial abstract domain with parameters $p_1,\allowbreak\dots,\allowbreak p_n$, there is a family of formulas $F_1,\dots,F_n$ that uniquely defines the optimal parameters $p'_1,\dots,p'_{n'}$ of the least inductive invariant for that loop, with respect those $p_1,\dots,p_n$ in the precondition (the free variables of $F_i$ are among $p_1,\dots,p_n,p'_i$).
\end{itemize}

The main obstacle is the high cost of quantifier elimination in the theory of real closed fields.
The other crucial difference is that it is in general impossible to move from such a formula to a formula computing $p'_i$ from $p_1,\dots,p_n$, as we did in \S~\ref{part:template_linear_constraints_gen}. By performing the cylindrical algebraic decomposition with variables $p_1,\dots,p_n$ first, we could obtain the tree structure with case disjunctions, as the output of Algorithm~\ref{alg:ToIfThenElseTree}. But at the leaves, we would obtain formulas defining $p'_i$ as a specific root of a polynomial in the variable $p'_i$, with coefficients themselves polynomials in $p_1,\dots,p_n$.

In Sec.~\ref{part:template_linear_constraints_gen}, we explained how to turn a formula over linear arithmetic defining a partial function from $p_1,\dots,p_n$ to $p'_i$ --- that is, a relation between $(p_1,\dots,p_n,p'_i)$ such that for any choice of $p_1,\dots,p_n$ there is at most one suitable $p'_i$ --- into an algorithmic function, expressed using linear assignments and if-then-else over linear inequalities.
Can we do the same here? That is, can we turn a formula over nonlinear arithmetic defining a partial function from $p_1,\dots,p_n$ to $p'_i$ into a simple algorithm written using, say, if-then-else tests and normal arithmetic operators as well as the $n$-th root operations $\sqrt[n]{}$~? For instance, if we have a formula $p' \geq 0 \land {p'}^2 = 2p$, we would like to obtain $p' = \sqrt{2p}$.

Unfortunately, this is impossible in the general case.
The Abel-Ruffini theorem, from Galois theory, states that for polynomials in one variable of degrees higher or equal to 5, there is in general no way to express the value of the roots using only arithmetic operations ($+$, $-$, $\times$, $/$) and radicals ($\sqrt[n]{}$). Thus, we cannot hope to obtain in general a simple algorithm expressing $p'_i$ as a function of $p_1, \dots, p_n$ using tests, arithmetic operations ($+$, $-$, $\times$, $/$) and radicals ($\sqrt[n]{}$).

Let us now assume that there are no precondition parameters $p_1,\dots,p_n$ or, equivalently, that we know exactly their value. Would it be at least possible to compute the values of the $p'_i$?

$p'_i$ is defined as the only solution of a logical formula with a single free variable, built using polynomial arithmetic. By putting this formula into DNF, we can reduce the problem to computing the only solution, if any, of a conjunction of polynomial inequalities and equalities. Such a solution can be computed to arbitrary precision\,; in fact, for any $\epsilon$, one can obtain bounds $[l,h]$ such that $l \leq p'_i \leq h$ and $h-l \leq \epsilon$. Unfortunately, the cost of such computations is high.~\cite{Basu_Pollack_Roy_algoreal_2006}

\subsubsection{Experiments}
Consider $l_x \leq x \leq h_x$, $l_y \leq y \leq h_y$ and the problem of generating the optimal abstract transfer function for the multiplication operation, $z := x * y$. We wish to obtain $l_z$ and $h_z$ such that $l_z \leq z \leq h_z$, and $l_z$ and $h_z$ are optimal ($l_z$ is maximal, $h_z$ is minimal). We first define the set of admissible (not necessarily minimal)~$h_z$:
\begin{equation}
A \definedAs \forall x \forall y (l_x \leq x\leq h_x \land l_y \leq y\leq h_y
  \Rightarrow x y\leq h_z)
\end{equation}

Now we define the \emph{least} value for $h_z$:
\begin{equation}
O \definedAs A \land \forall h~(A[h/h_z] \Rightarrow h \geq h_z)
\end{equation}
The free variables of this formula are $l_x$, $h_x$, $l_y$, $h_y$ and~$h_z$.

\begin{figure}
{\small
\begin{multline}
\left(l_y<0\land
  \left(\left(h_y=l_y\land h_x\geq \
    \frac{l_x l_y}{h_y}\land h_z=l_x \
    l_y\right) \right.\right.\\\left.\left.
  \lor (l_y<h_y\leq 0\land ((l_x\leq 0\land \
    h_x\geq l_x\land h_z=l_x l_y)
  \lor \
     (l_x>0\land h_x\geq l_x\land h_z=h_y \
     l_x))) \right.\right. \\ \left.\left.
  \lor \left(h_y>0\land \left(\left(l_x<0\land \
     \left(\left(l_x\leq h_x\leq \frac{l_x \
       l_y}{h_y}\land h_z=l_x l_y\right)\lor \
     \left(h_x>\frac{l_x l_y}{h_y}\land \
     h_z=h_x h_y\right)\right)\right) 
     \right.\right.\right.\right. \\ \left.\left.\left.\left.
  \lor (l_x=0\land \
     h_x\geq 0\land h_z=h_x h_y)\lor (l_x>0\land \
     ((h_x=l_x\land h_z=h_y l_x) \
     \right.\right.\right.\right. \\ \left.\left.\left.\left.
     \lor (h_x>l_x\land h_z=h_x \
     h_y)))\right)\right)\right)\right)\\
\lor (l_y=0\land \
((h_y=0\land h_x\geq l_x\land h_z=0)\lor \
(h_y>0\land ((l_x<0\land ((l_x\leq h_x\leq 0\land \
h_z=0) \\
\lor (h_x>0\land h_z=h_x h_y)))\lor
(l_x=0\land h_x\geq 0\land h_z=h_x h_y)\lor \
(l_x>0\land ((h_x=l_x\land h_z=h_y l_x) \\
\lor (h_x>l_x\land h_z=h_x \
h_y)))))))\\
\lor \left(l_y>0\land \
\left(\left(h_y=l_y\land \left(\left(l_x<0\land \
\left(\left(h_x=\frac{l_x l_y}{h_y}\land \
h_z=l_x l_y\right)
\right.\right.\right.\right.\right.\right. \\ \left.\left.\left.\left.\left.\left.
\lor \left(h_x>\frac{l_x \
l_y}{h_y}\land h_z=h_x \
h_y\right)\right)\right)
\lor (l_x=0\land h_x\geq 0\land \
h_z=h_x h_y)\lor
\right.\right.\right.\right. \\ \left.\left.\left.\left.
\left(l_x>0\land \
\left(\left(h_x=\frac{l_x l_y}{h_y}\land \
h_z=l_x l_y\right)\lor \left(h_x>\frac{l_x \
l_y}{h_y}\land h_z=h_x \
h_y\right)\right)\right)\right)\right)
\right.\right. \\ \left.\left.
\lor (h_y>l_y\land \
((l_x<0\land ((h_x=l_x\land h_z=l_x \
l_y)\lor (l_x<h_x\leq 0\land h_z=h_x \
l_y)
\right.\right. \\ \left.\left.
\lor (h_x>0\land h_z=h_x h_y)))\lor \
(l_x=0\land h_x\geq 0\land h_z=h_x h_y) \
\right.\right. \\ \left.\left.
\lor (l_x>0\land ((h_x=l_x\land h_z=h_y \
l_x)
\lor (h_x>l_x\land h_z=h_x \
h_y)))))\right)\right).
\end{multline}
}

\caption{This formula is the result of quantifier elimination. It defines $h_z$ to be the least upper bound of $xy$ for $x \in [l_x,h_x]$ and $y \in [l_y,h_y]$.}
\label{fig:interval_mul}
\end{figure}

Mathematica~7 performs quantified elimination by cylindrical algebraic decomposition on this formula in 4.3~s and yields a large formula (Fig.~\ref{fig:interval_mul}) with many case disjunctions, most notably on the sign of $l_x$, $h_x$, $l_y$, $h_y$. This is quite natural: the monotonicity of the function $y \mapsto xy$ changes according to the sign of~$x$. This function is equivalent to the much more terse
\begin{equation}
h_z = \max(l_xl_y, h_xl_y, l_xh_y, h_xh_y)
\end{equation}

This illustrates the limit of our approach on nonlinear problems: even on simple program constructions and with simple invariants, quantifier elimination takes nonneglible time and outputs complicated formulas. We therefore did not pursue this direction further.

\section{Related work}
Since the first numerical abstract analysis techniques were proposed in the 1970s, there has been considerable work on improving precision, efficiency, or both. Without attempting to be exhaustive, we shall now describe a few of the approaches and how they differ from ours.



\label{part:related}

\subsection{Relational abstract domains and modular analysis}
There is a sizeable amount of literature concerning relational numerical abstract domains; that is, domains that express constraints between numerical variables. Convex polyhedra were proposed in the 1970s~\cite{Halbwachs_PhD,CousotHalbwachs78}, and there have been since then many improvements to the technique; a bibliography was gathered by~\citet{PPL}. Algorithms on polyhedra are costly and thus a variety of domains intermediate between simple interval analysis and convex polyhedra were proposed~\cite{Mine_AST_WCRE01,Clariso_Cortadella_SAS2004,Sankaranarayana+others/05/Scalable}. 

It is possible to use relational abstract domains such as polyhedra to model the input/output relationship of a program, function or block~\citep[\S7.2, p.~112]{Halbwachs_PhD}. Instead of considering only the current values of the program variables $(v_1,\dots,v_n)$ at the various program points, one also considers the initial values $(v^0_1,\dots,v^0_n)$ of these variables at the beginning of the program, function or block; thus the computed polyhedra, for instance, relate $(v^0_1,\dots,v^0_n,v_1,\dots,v_n)$. We employed this approach when dealing with recursive procedures (\S\ref{part:McCarthy91}). Such an approach is modular: one can for instance analyse a procedure in such a way, and plug the result of the analysis at each point of call; though of course one loses optimality. It also provides for modular analysis of loop nests: one first analyses the innermost loop, and then replace this innermost loop by the result of the analysis, considered as a nondeterministic program; one then proceeds to the next innermost loop.

One limit of this approach is that the relationship between input and output is constrained by the abstract domain. Most numerical abstract domains concern \emph{convex} relations: difference bound constraints, octagons, polyhedra etc. are all geometrically convex (given two points $a,b$ in the concretisation, the segment $[a,b]$ is also in the concretisation). Note that the result of the analysis of the absolute value function (\S\ref{part:template_linear_constraints_transformers}), as expressed by Rel.~\ref{eqn:abs_dnf}, or that of the rate limiter (\S\ref{part:rate_limiter}), are piecewise linear but not convex.

The idea of producing procedure summaries \cite{Sharir_Pnueli_81} as formulas mapping input bounds to output bounds is not new. \citet{Rugina_Rinard_05}, in the context of pointer analysis (with pointers considered as a base plus an integer offset), proposed a reduction to linear programming. This reduction step, while sound, introduces an imprecision that is difficult to measure in advance; our method, in contrast, is guaranteed to be ``optimal'' in a certain sense. \citeauthor{Rugina_Rinard_05}'s method, however, allows some nonlinear constructs in the program to be analysed. \citet{DBLP:conf/cc/MartinAWF98} proposed applying interprocedural analysis to loops.

\citet{Seidl_ESOP07} also produce procedure summaries as numerical constraints. Our procedure summaries are implementations of the corresponding abstract transformer over some abstract domain, while theirs outputs a relationship between input and output concrete values. Their analysis considers a \emph{convex} set of concrete input-output relationships, expressed as \emph{simplices}, a restricted class of convex polyhedra. This restriction trades precision for speed: the generator and constraint representations of simplices have approximately the same size, while in general polyhedra exponential blowup can occur. Tests by arbitrary linear constraints cannot be adequately represented within this framework. \citet[Sec.~4]{Seidl_ESOP07} propose deferring those constraints using auxiliary variables; this, however, loses some precision. Their analysis and ours are therefore incomparable, since they make different choices between precision and efficiency.

\citet{Reps_CAV05} proposed an interprocedural analysis of numerical properties of functions using weighted pushdown automata. The ``weights'' are taken in a finite height abstract domain, while the domains we consider have infinite height.

\subsection{Computations of exact fixed points}
The limitations of the widening approach explained in \S\ref{part:absint} have been recognized for long. There has therefore been extensive research about computing precise inductive invariants, if possible the least inductive invariant inside the abstract domain considered.

Several methods have been proposed to synthesize invariants without using widening operators~\cite{Colon_CAV03,Cousot05-VMCAI,Sankaranarayanan_SAS04}. In common with us, they express as constraints the conditions under which some parametric invariant shape truly is an invariant, then they use some resolution or simplification technique over those constraints. Again, these methods are designed for solving the problem for one given set of constraints on the inputs, as opposed to finding a relation between the output or fixed-point constraints and the input constraints. In some cases, the invariant may also not be minimal.

\citet{DBLP:conf/sas/BagnaraHMZ05,DBLP:journals/scp/BagnaraHRZ05} proposed improvements over the ``classical'' widenings on linear constraint domains~\cite{Halbwachs_PhD}. \citet{DBLP:conf/cav/GopanR06} introduced ``lookahead widenings'': standard widening-based analysis is applied to a sequence of syntactic restrictions of the original program, which ultimately converges to the whole program; the idea is to distinguish phases or modes of operation in order to make the widening more precise.
\citet{DBLP:conf/sas/GonnordH06,Gonnord_PhD,LEROUX-SUTRE-SAS2007} have proposed \emph{acceleration} techniques: when the transition relation $\tau$ is of certain particular forms, it is possible to compute its transitive closure $\tau^+$ exactly or with small imprecision. Typically, acceleration techniques have difficulties dealing with programs where the control flow is not \emph{flat}, for instance when there are paths through a loop body that affect the iteration variables in different ways, such as the circular buffer example (Listing~\ref{ex:circularbuffer}).

\citet{arXiv:0806.1160,GGTZ:07,DBLP:conf/cav/CostanGGMP05} proposed a ``policy iteration'' or ``strategy iteration'' approach,%
\footnote{The terminology comes from game theory. In broad terms, their consider equations with max/min operators, which are similar to the ``minimax'' operators appearing in the definition of the value of games in game theory. Choosing which argument of a ``min'' is used corresponds, in game theory terms, of choosing a \emph{strategy} or \emph{policy} for a player that tries to minimise the value of the game.}
by downwards iterations providing successive over-ap\-pro\-xi\-ma\-tions of the least fixed point. Their approach can fail to converge to the least fixed point, for instance with expansive semantics such as those of the ``circular buffer'' example (Listing~\ref{ex:circularbuffer}), though for some classes of programs it converges to the least fixed point~\citep{arXiv:0806.1160}.

\citet{Gawlitza_Seidl_ESOP07} proposed another policy iteration approach, which is guaranteed to provide the least fixed point of the system of abstract equations. In contrast to the above method, they use upwards iterations, so each value computed is an under-approximation of the abstract least fixed point. They extended that approach to template linear constraint domains \citep{DBLP:conf/csl/GawlitzaS07}. The differences with our approach are twofold:
\begin{itemize}
\item Their approach computes the least fixed point of a system of min/max abstract equations, derived from the source code of the program. In intuitive terms, min's correspond to conditions (and closures operators in relational domains), and max's to ``merge points'' in the control flow graph (end of if-then-else). This approach thus incurs the same problem of ``undistinguished paths'' as the example from \S\ref{part:absint}. Even then, the policy iteration algorithm may iterate across a number of iterations exponential in the number of merge points.

An alternative would be to consider a cut-set for the control flow graph and distinguish each path between two points in the cut-set. in other words, for a single loop, one would consider each individual control path inside the loop body. The number of such paths is exponential in the number of tests, and thus the ``max'' operation in the abstract equations would also have an exponential number of arguments. Such an approach would compute the same result as ours; however, it would be unworkable without further work to get rid of the explicitly exponential number of arguments in the ``max'' operation.

\item Their approach needs all preconditions exactly known. In contrast, we compute it as an explicit function of the precondition. In short, our invariants are \emph{parametric} in the precondition, while theirs are not.
\end{itemize}

\citet{Gulwani_PLDI08} have also proposed a method for generating linear invariants over integer variables, using a class of templates. The methods described in the present article can be applied to linear invariants over integer variables in two ways: either by abstracting them using rationals (as in examples in Sec.~\ref{part:loop_example}, \ref{part:noop_nests}), either by replacing quantifier elimination over rational linear arithmetic by quantifier elimination over linear integer arithmetic, also known as Presburger arithmetic (\S\ref{part:integers}). \citeauthor{Gulwani_PLDI08} instead chose to first consider integer variables as rationals, so as to be able to compute over rational convex polyhedra, then bound variables and constraint parameters so as to model them as finite bit vectors, finally obtaining a problem amenable to SAT~solving. Program variables \emph{are} finite bit vectors in most industrial programming languages, and parameters to useful invariants over integer variables are often small, thus their approach seems justified. We do not see, however, how their method could be applied to programs operating over real or floating-point variables, which are the main motivation for the present article.

\subsection{Limitations of template-based approaches}
Much, if not all, of the published results on computing least inductive invariants in abstract domains, or at least abstract fixed points, deal with template domains \cite{GGTZ:07,DBLP:conf/csl/GawlitzaS07,Gulwani_PLDI08}, including intervals~\cite{DBLP:conf/cav/CostanGGMP05,Gawlitza_Seidl_ESOP07}. The fundamental reasons for this are:
\begin{itemize}
\item The domain of convex polyhedra is not closed under infinite intersection. Thus, in general, there is no best abstraction of a set of states. In general, there is no least inductive invariant inside the domain.
\item These methods replace the problem of dealing with arbitrarily complex shapes such as convex polyhedra by dealing with a vector of real numbers in finite, fixed dimension. The coefficients of these vectors are then amenable to a variety of constraint solving techniques.
\end{itemize}

A common criticism of template-based approaches, is that they suppose that one knows the interesting templates beforehands --- in the case of linear constraint domains, interesting directions in space. In contrast, methods based on general convex polyhedra infer these directions themselves, through the convex hull and widening operations~\citep{CousotHalbwachs78,Halbwachs_PhD}. For instance, an analysis using standard polyhedra of the following program will infer that $2x=y$, while a template based approach would succeed in doing so only if a template of the form $2x-y$ has been provided:
\lstset{language=Cext}
\begin{lstlisting}
x = y = 0;
while (true) {
  x = x+1;
  y = y+2;
}
\end{lstlisting}

We understand and share that criticism. In some cases there exist ``natural'' templates, such as intervals, or when dealing with timing or scheduling constraints, difference bounds $v_i - v_j \leq C$, but in general, finding the correct templates seems a hard problem. A suggestion is to run iterations using general convex polyhedra and look at the stable directions of the faces of these polyhedra.

We think however that criticism of template domains should be part of wider considerations on how to choose the proper abstract domain.
Abstract interpretation essentially replaces the unsolvable problem of computing the least inductive invariant of the concrete problem by the solvable problem of computing an inductive invariant in an abstract domain --- in some lucky cases such as the one dealt with in this article, actually obtaining the least inductive invariant \emph{in the abstract domain}. The choice of the abstract domain is at present somewhat arbitrary, typically hardwired into the analysis tool, at best chosen by some command-line flags.

Convex polyhedra \citep {Halbwachs_PhD,CousotHalbwachs78,halbwachs94b} are popular because many programming idioms naturally exhibit convexity --- for instance, the set of loop indices $(i,j,k,\dots)$ of nested loops occurring in numerical analysis programs is often convex. Yet, one can easily think of programs where some interesting properties are not convex (a test $|x| \geq 1$, for instance). There are some cases where it is important not to enforce convexity and instead implement disjunctive domains, capable of representing properties such as $x \leq -1 \vee x \geq 1$; an example is \emph{trace partitioning}~\cite{Rival_Mauborgne_TOPLAS07}. Thus, a static analysis tool based on general convex polyhedra also enforces an \emph{a priori} convex shape that might not be representative of the useful program invariants.

While convex polyhedra are not enough, they are often ``too much'': their complexity is too high for many applications. Indeed, even for less costly constraint domains such as octagons \citep{mine:hosc06,Mine_AST_WCRE01}, it is simply too expensive to compute constraints between all visible program variables, so some analysers choose \emph{a priori} to consider relations only between certain variables --- an approach knowing as \emph{packing} in the Astr\'ee tool~\citep{Monniaux_ASIAN06,ASTREE_ESOP05}. Again, the packing choice is a form of \emph{a priori} template guess made by the analyser, using heuristics that look at the way the program is organized.

Further research is obviously needed in how to choose and adapt abstract domains so that they can represent interesting inductive properties at reasonable costs. This problem is similar to the problem of finding the correct predicates in predicate abstraction. For this, various methods for finding predicates in addition to those syntactically present in the program have been proposed, especially those based on the analysis of spurious counterexamples (\emph{counterexample-guided abstraction refinement} or CEGAR).\footnote{The literature on CEGAR is too vast to be cited here without unfairness. \citet{Clarke_et_al_JACM03} introduced this approach for symbolic model checking. Notable applications to software model checking include the \textsc{Blast} \citep{BLAST_STTT07} and \textsc{Slam} \citep{SLAM_PLDI01} tools.} Perhaps similar ideas could be employed for suggesting suitable additional numerical templates, finer-grained packing or disjunctions.

\subsection{Relational domains beyond polyhedra}
In earlier works, we have proposed a method for obtaining input-output relationships of digital linear filters with memories, taking into account the effects of floating-point computations~\cite{Monniaux_CAV05}. This method computes an exact relationship between bounds on the input and bounds on the output, without the need for an abstract domain for expressing the local invariant; as such, for this class of problems, it is more precise than the method from this article. This technique, however, cannot be easily generalized to cases where the operator block contains tests and other nonlinear constructs; the semantics of nonlinear constructs must be approximated by e.g. interval analysis.

There have been several published approaches to finding nonlinear relationships between program variables. One approach obtains polynomial equalities through computations on ideals using Gr\"obner bases~\citep{Rodriguez_Kapur_SAS2004,1222597}. This work only deals with equalities (not inequalities), uses a classical approach of computing output constraints from a set of input constraints (instead of finding relationships between the two sets of constraints), and deals with loops using a widening operator. In comparison, our approach abstracts whole program fragments, and is modular --- it is possible to ``plug'' the result of the analysis of a procedure at the location of a procedure call, though of course this is less precise than inlining the procedure.

Since nonlinear relations are notoriously costly to compute upon, \citet{DBLP:conf/sas/BagnaraRZ05} have proposed using further abstraction to be able to reduce the problem to computations over convex polyhedra.

\citet{Kapur_ACA04} also proposed to use quantifier elimination to obtain invariants: he considers program invariants with parameters, and derives constraints over those parameters from the program.
Our work improves on his by noting that least invariants of the chosen shape can be obtained, not just any invariant; that the abstraction can be done modularly and compositionally (a program fragment can be analysed, and the result of its analysis can be plugged into the analysis of a larger program), or combined into a ``conventional'' abstract interpretation framework (by using invariants of a shape compatible with that framework), and that the resulting invariants can be ``projected'' to obtain numerical quantities.

\section{Conclusion and future prospects}
\label{part:conclusion}
Writing static analysers by hand has long been found tedious and error-prone. One may of course prove an existing analyser correct through assisted proof techniques, which removes the possibility of soundness mistakes, at the expense of much increased tediousness. In this article, we proposed instead effective methods to synthesize abstract domains by automatic techniques. The advantages are twofold: new domains can be created much more easily, since no programming is involved; a single procedure, testable on independent examples, needs be written and possibly formally proved correct. To our knowledge, this is the first effective proposal for generating numerical abstract domains automatically, and one of the few methods for generating numerical summaries. Also, it is also the only method so far for computing summaries of \emph{floating-point} functions.

We have shown that floating-point computations could be safely abstracted using our method. The formulas produced are however fairly complex in this case, and we suspect that further over-approximation could dramatically reduce their size. There is also nowadays significant interest in automatizing, at least partially, the tedious proofs that computer arithmetic experts do and we think that the kind of methods described in this article could help in that respect.

We have so far experimented with small examples, because the original goal of this work was the automatic, on-the-fly, synthesis of abstract transfer functions for small sequences of code that could be more precise than the usual composition of abstract of individual instructions, and less tedious for the analysis designer than the method of pattern-matching the code for ``known'' operators with known mathematical properties. A further goal is the precise analysis of longer sequences, including integer and Boolean computations. We have shown in Sec.~\ref{part:partitioning} how it was possible to partition the state space and abstract each region of the state-space separately; but naive partitioning according to $n$ Booleans leads to $2^n$ regions, which can be unbearably costly and is unneeded in most cases. We think that automatic refinement and partitioning techniques \cite{jeannet03a} could be developed in that respect.

The main practical application that we envision is to be able to analyse numerical operator blocks from synchronous programming languages such as \textsc{Simulink},%
\footnote{\textsc{Simulink} is a graphical dataflow modeling tool sold as an extension to the \textsc{Matlab} numerical computation package. It allows modeling a physical or electrical environment along the computerized control system. A code generator tool can then provide executable code for the control system for a variety of targets, including generic~C. \textsc{Simulink} is available from \href{http://www.mathworks.com/}{The Mathworks}.}
\textsc{Scicos},%
\footnote{\href{http://www.scicos.org/}{\textsc{Scicos}} is a graphical dataflow modeling tool coming with the \textsc{Scilab} numerical computation package, similar in use to \textsc{Simulink}.~\citep{Scilab_Scicos} It is available from INRIA under the GNU General Public License and also has code generation capabilities.}
\textsc{Lustre},%
\footnote{\textsc{Lustre} is a synchronous programming language, from which code can be generated for a variety of platforms \citep{LUSTRE_POPL87}.}
\textsc{Scade}%
\footnote{\textsc{Scade} is a graphical synchronous programming language derived from \textsc{Lustre}. It is available from \href{http://www.esterel-technologies.com/}{Esterel Technologies}. It was used for implementing parts of the Airbus A380 fly-by-wire systems, among others. \cite{DBLP:conf/safecomp/SouyrisD07,DBLP:conf/sas/DelmasS07}}
or \textsc{Sao},%
\footnote{\textsc{Sao} is an earlier industrial graphical synchronous programming language, used, for implementing parts of the Airbus A340 fly-by-wire systems \citep{Briere_Traverse_FTCS-23}, among others.}
which are widely used for programming control systems \citep{Astrom_Wittenmark_97}, particularly in the automative and avionic industries. In order to obtain good analysis precision, such blocks often have to be analysed as a whole instead of decomposing them into individual components and applying individual transfer functions, as in our rate limiter example. The static analysis tool \textsc{Astr\'ee} \citep{Monniaux_ASIAN06,ASTREE_TASE07,ASTREE_ESOP05,ASTREE_PLDI03,DBLP:conf/safecomp/SouyrisD07,DBLP:conf/sas/DelmasS07} outputs few, if any, false alarms on some classes of control programs because it has specific specialized transfer functions for certain operator blocks or coding patterns. Such transfer functions had to be implemented by hand; the techniques described in the present article could have been used to implement some of them automatically and even on-the-fly.

There are two important drawbacks to our method, which make it currently only useful for very precise analysis of small parts of programs. The first is that we need to ``see'' the whole of the loop or function that we are analysing, the instructions of which must belong to the class of constructs that we are capable of dealing with, or at least can be abstracted by them. In contrast, iterative techniques are more tolerant: they see the program state locally, at each program point, and the numerical analysis may easily interact with other analyses, such as pointers~\citep{ASTREE_PLDI03,BlanchetCousotEtAl02-NJ}. The second issue is the high cost of quantifier elimination. Despite our work on new algorithms \citep{Monniaux_LPAR08}, in which we are still making progress, scalability remains an issue.

\section*{Acknowledgements}
We would like to thanks the anonymous referees for careful reading.



\newcommand{\doix}[1]{\href{http://dx.doi.org/#1}{#1}\endgroup}
\newcommand{\doi}{\begingroup\footnotesize doi: \catcode`\_=13\def\_{\textunderscore}\doix}
\newcommand{\isbn}[1]{{\footnotesize \href{http://www.worldcat.org/isbn/#1}{ISBN #1}}}
\newcommand{\issn}[1]{{\footnotesize \href{http://www.worldcat.org/issn/#1}{ISSN #1}}}
\renewcommand{\UrlFont}{\footnotesize\tt}

\bibliographystyle{dmplainnat}
\bibliography{automatic_modular_exact_invariants_lmcs}
\end{document}